\documentclass[11pt]{article}
\usepackage[T1]{fontenc}
\usepackage{lmodern,microtype}
\usepackage[left=1.15in,right=1.15in,top=1.2in,bottom=1.2in]{geometry}
\usepackage{amsmath,amssymb,amsthm,mathtools}
\usepackage{xcolor,tikz}
\usepackage{booktabs,tabularx,needspace}
\usetikzlibrary{decorations.pathreplacing}
\usepackage[section]{placeins}
\usepackage[most,breakable]{tcolorbox}
\usepackage{authblk}

\usepackage[colorlinks=true,linkcolor=blue!45!black,citecolor=blue!45!black,urlcolor=blue!45!black]{hyperref}
\newtheorem{fact}{Fact}
\newtheorem{lemma}{Lemma}
\newtheorem{theorem}{Theorem}
\newtheorem{corollary}{Corollary}
\theoremstyle{remark}

\theoremstyle{plain}

\definecolor{mainresultaccent}{HTML}{176B87}
\definecolor{technicalaccent}{HTML}{657A45}
\tcbset{resultbox/.style={
 enhanced, unbreakable, boxrule=0.5pt, arc=1.5pt, outer arc=1.5pt,
 left=5.75pt, right=5.75pt, top=4pt, bottom=4pt,
 grow sidewards by=10pt, before skip=8pt, after skip=8pt
}}
\tcolorboxenvironment{theorem}{resultbox,
 colback=mainresultaccent!9!white, colframe=mainresultaccent!38!white}
\tcolorboxenvironment{lemma}{resultbox,
 colback=technicalaccent!6!white, colframe=technicalaccent!32!white}

\newcommand{\Tr}{\operatorname{tr}}
\newcommand{\id}{\operatorname{id}}
\newcommand{\ran}{\operatorname{ran}}
\newcommand{\N}{\mathcal N}
\newcommand{\Ic}{I_{\mathrm c}}
\newcommand{\Enc}{\mathcal E}
\newcommand{\Dec}{\mathcal D}
\newcommand{\I}{\mathbf{1}}
\newcommand{\Hmin}{H_{\min}}
\newcommand{\ket}[1]{\lvert #1\rangle}
\newcommand{\bra}[1]{\langle #1\rvert}
\newcommand{\ketbra}[2]{|#1\rangle\!\langle#2|}
\newcommand{\braket}[2]{\langle#1|#2\rangle }
\newcommand{\op}[1]{\lVert #1\rVert_\infty}
\newcommand{\pn}[1]{\lVert #1\rVert_\pi}
\newcommand{\norm}[1]{\lVert #1\rVert}

\allowdisplaybreaks[1]

\hypersetup{pdftitle={Strong Converses for Quantum Channel Capacities from Blowing-Up Lemmata},pdfauthor={Salman Beigi and Marco Tomamichel}}

\title{Strong Converses for Quantum Channel Capacities from Blowing-Up Lemmata}
\author[1,2]{Salman Beigi}
\author[2,3]{Marco Tomamichel}
\affil[1]{\emph{School of Mathematics, Institute for Research in Fundamental Sciences (IPM), P.O. Box 19395-5746, Tehran, Iran}}
\affil[2]{\emph{Centre for Quantum Technologies, National University of Singapore, Singapore 117543, Singapore}}
\affil[3]{\emph{Department of Electrical and Computer Engineering, National University of Singapore, Singapore 117583, Singapore}}
\date{29 September 2026}

\begin{document}
\maketitle

\begin{abstract}
A quantum channel can transmit quantum states, classical messages, or classical messages that remain secret from its environment. We prove exponential strong converses for all three tasks over arbitrary finite-dimensional memoryless quantum channels, at their respective regularized capacities and allowing arbitrary entangled inputs across channel uses and joint decoding. Above quantum or classical capacity, entanglement fidelity or decoding success probability, respectively, decays exponentially with the number of channel uses. Above private capacity, fidelity to an ideal shared secret key decays exponentially; in particular, vanishing secrecy error forces decoding success to vanish. The proofs combine one-shot quantum blowing-up lemmas, which convert codes with low success into reliable codes with a controlled loss in code size, with a low-degree polynomial approximation of the channel's Stinespring image projector. The loss is governed by the projective tensor norm of the approximation across the receiver--environment bipartition. For memoryless channels, the polynomial construction gives exponential accuracy while keeping the rate loss arbitrarily small, turning the corresponding weak converses into strong converses. 
\end{abstract}

\section{Introduction}\label{sec:introduction}
A quantum channel can transmit quantum states, ordinary classical messages, or classical messages that remain secret from its environment. Each task has a capacity: the largest rate achievable with errors tending to zero. A weak converse excludes asymptotically perfect communication above capacity, but leaves open whether a nonzero fidelity or probability of success can survive. A strong converse makes capacity a sharp threshold; an exponential strong converse quantifies how rapidly the relevant measure of success vanishes.

In this work, we establish exponential strong converses for all three tasks over arbitrary finite-dimensional memoryless quantum channels. The results hold at their regularized capacities and allow arbitrary entanglement across channel uses and joint decoding. Our common approach combines one-shot quantum blowing-up lemmas, which convert codes with low success into reliable codes with a controlled loss in size, with a low-degree approximation of the channel's Stinespring image projector. For private communication, the conversion controls reliability and secrecy simultaneously, and the exponentially decaying quantity is the fidelity to an ideal shared secret key.

For quantum communication, the sender encodes an unknown quantum state across many uses of the channel, and the receiver applies a joint decoding operation. Successful transmission must preserve not only the state itself but also its entanglement with an inaccessible reference. For an $M_n$-dimensional system transmitted through $n$ independent uses, the rate is $(\log M_n)/n$ qubits per channel use. The quantum capacity $Q(\N)$ is the supremum of rates achievable with entanglement-transmission fidelity tending to one.

The Lloyd--Shor--Devetak theorem identifies this capacity with regularized coherent
information~\cite{Lloyd,Shor,Devetak}:
\begin{equation}\label{eq:capacityformula}
 Q(\N)=\sup_{n\ge1}\frac1n Q^{(n)}(\N),\qquad
 Q^{(n)}(\N)=\max_{\rho_{A^n}}\Ic(\rho_{A^n},\N^{\otimes n}),
\end{equation}
where $\Ic(\rho,\N)=H(\N(\rho))-H(\N^c(\rho))$, $H$ denotes the von Neumann entropy, and $\N^c$ is a
complementary channel. Regularization permits input states entangled across
arbitrarily many channel uses. The corresponding weak converse~\cite{Devetak} excludes fidelity tending to one above $Q(\N)$; our exponential strong converse forces it to vanish exponentially.

For classical discrete memoryless channels, strong converses make capacity
a sharp threshold for reliable communication. Arimoto's converse gives
exponential decay of the probability of correct decoding whenever the rate
exceeds capacity~\cite{Arimoto}. Ogawa and Nagaoka~\cite{OgawaNagaoka} and Winter~\cite{WinterCoding} also established the strong converse for the transmission of classical information through classical--quantum channels. These results, however, cover only
product-state coding through quantum channels. Allowing arbitrary entangled
codewords is a distinct problem; strong converses in that setting have been
proved for some classes of channels, including the covariant examples
treated by K\"onig and Wehner~\cite{KoenigWehner}.
With unlimited shared entanglement, Gupta and Wilde proved an exponential
strong converse for the entanglement-assisted classical capacity of every
finite-dimensional quantum channel~\cite{GuptaWilde}.

For unassisted quantum communication, the strong-converse problem at the
regularized coherent-information capacity has remained open in general.
A major step was the \emph{pretty strong} converse of Morgan and
Winter~\cite{MW} for degradable channels. They showed that, above
capacity, the asymptotic fidelity is bounded away from one, without proving
that it vanishes. They also reduced the full strong converse for degradable
channels to a corresponding statement for symmetric zero-capacity channels.
Tomamichel, Wilde, and Winter proved an exponential decay rate above the channel's
Rains information, giving an exponential strong converse at capacity for
generalized dephasing channels~\cite{TWW}. More recently, Kondra, Brinster, Kampermann, Bru{\ss}, and
Wyderka~\cite{Kondra} proved exponential strong converses for every
finite-dimensional degradable and antidegradable channel. Their
antidegradable-channel bound, together with the Morgan--Winter reduction,
establishes the degradable case. These results settle the problem for two
central classes of channels, although a general channel need have neither property.

Here we establish an exponential strong converse for arbitrary
finite-dimensional memoryless quantum channels. We formulate the result
for entanglement-generation codes: an arbitrary encoded state $\rho_{RA^n}$, with $\dim R=M_n$, is sent
through $\N^{\otimes n}$ on $A^n$, followed by a decoder to an
$M_n$-dimensional target space $\widehat L$. The fidelity $F_n$ is the overlap with
the normalized maximally entangled state on $R\widehat L$, in the
squared-fidelity convention. This includes entanglement-transmission codes,
and allows mixed encoded states and arbitrary joint decoders.

\begin{theorem}[Exponential strong converse]\label{thm:strongconverse}
For every finite-dimensional channel $\N$ and every rate gap $\gamma>0$,
there is $\alpha_\gamma>0$ such that every sequence of $(n,M_n)$ codes with
\begin{equation}\label{eq:main-rate}
 \frac1n\log M_n\ge Q(\N)+\gamma
\end{equation}
satisfies $F_n<2^{-\alpha_\gamma n}$ for all sufficiently large $n$.
\end{theorem}

The proof has two main ingredients: a fully quantum blowing-up lemma and a
polynomial approximation of the orthogonal projection on the image of the Stinespring dilation of the channel. The first turns
an approximation with a controlled projective tensor norm into a conversion
from low-fidelity to high-fidelity codes. The second constructs an approximation
whose accuracy and norm make this conversion useful at rates above the capacity.

Classical blowing up provides a route from weak to strong converses. The lemma
of Ahlswede, G\'acs, and K\"orner~\cite{AGK}, with an information-theoretic proof
due to Marton~\cite{Marton}, enlarges events of nonnegligible probability under a product distribution
to a neighbourhood that has a high probability. 
Our fully quantum blowing-up lemma, stated as a one-shot code-conversion theorem in
Section~\ref{sec:blowingup}, constructs a high-fidelity code from a low-fidelity
one, with a controlled loss in the dimension of the code. Curiously, the loss is governed by the
projective tensor norm of an approximation to the orthogonal projection on the image of the channel's Stinespring dilation across the receiver--environment bipartition. This norm measures the cost of
expressing the operator as a sum of products of local operators. 

Earlier quantum blowing-up proposals do not directly give this code conversion.
Winter bounded the dimension of subspace neighbourhoods generated by local
operators, while posing probability amplification for classical--quantum outputs
as a conjecture~\cite[Lemma~II.24 and Conjecture~II.25]{WinterThesis}. The bipartite
lemma of Sreekumar, Hirche, Cheng, and Berta~\cite[Lemma~15, corrected version]{SHCB}
concerns product tests and tensor-power states; its alternative state must
commute with rank-one product projectors from suitable eigenbases of the null
marginals. General code outputs and decoder tests need not satisfy these
assumptions. Beigi, Datta, and Rouz\'e smooth tests relative to product states
via reverse hypercontractivity~\cite{BDR}, obtaining strong-converse bounds for
hypothesis testing and classical--quantum coding, but no entanglement-transmission
code conversion. The quantum Marton inequality of De Palma, Marvian, Trevisan,
and Lloyd~\cite[Theorem~2]{QuantumWasserstein} controls Wasserstein distance to a
product state, which alone does not guarantee high entanglement fidelity.
The stabilizer-code argument for Pauli channels~\cite{StabilizerBlowingUp} uses
a classical product distribution of errors, a representation unavailable for
general channels and codes.

The second ingredient supplies the operator approximation required by the
one-shot theorem. For a Stinespring isometry $U:A\to B\otimes E$ of the channel, its image
projector $\Pi=UU^\dagger$ specifies the allowed joint channel outputs.
Section~\ref{sec:approximation} provides an approximation of $\Pi^{\otimes n}$ by a low-degree
polynomial applied to the sum of the local complementary projectors $\I-\Pi_i$. This polynomial itself is constructed using
an approximation of the $\mathrm{NOR}$ function on the Boolean cube~\cite{Sherstov}.
Its degree controls both the error in operator norm and the projective tensor norm across $B^n:E^n$.
Low degree limits how many output pairs each term acts on, which keeps the
latter norm small. Both controls are needed: accuracy permits amplification
from low fidelity, while the projective norm bounds the loss in the number of transmitted
qubits.

Polynomial approximation has also been used in other parts of quantum
information theory. Kondra et al.~\cite{Kondra} approximate the Boolean
$\mathrm{NOR}$ function and use its Fourier coefficients to construct a
weight matrix controlling overlaps between different decoder placements.
Our use is closer to polynomial approximations of spectral projections.
Eldar and Harrow~\cite{EldarHarrow} use Chebyshev polynomials to establish
expansion properties of distributions generated by shallow quantum circuits,
which enter their results on local Hamiltonians with hard-to-approximate
ground states. Arad, Kitaev, Landau, and Vazirani~\cite{Arad}
construct approximate ground-space projectors from Chebyshev polynomials
of local Hamiltonians, balancing spectral suppression against entanglement
growth. The common principle is that low degree controls how an operator
acts across registers while allowing a sharp spectral approximation.
In the present work, it controls the projective tensor norm across the receiver and environment subsystems
and hence the rate loss in the fully quantum blowing-up argument.

Choosing the degree to be a sufficiently small linear fraction of $n$ makes
the approximation error exponentially small and keeps the rate loss below the
gap to capacity. A code whose fidelity decayed too slowly would then yield a
high-fidelity code above capacity, contradicting the existing weak converse.
Section~\ref{sec:strongconverse} gives the proof of
Theorem~\ref{thm:strongconverse}.

The same method also applies to the transmission of classical
information. Section~\ref{sec:classical-capacity} proves a quantum
blowing-up lemma for classical communication: a code with low success probability is converted into
a reliable code with a controlled loss in logarithmic message size.
The one-shot classical coding theorem of Renes and Renner~\cite{RenesRenner}
turns a smooth-entropy bound into this conversion. Combining it with the
polynomial approximation and the weak converse gives an exponential strong
converse at the regularized classical capacity~\cite{Holevo,SW}. In Section~\ref{sec:classical-capacity} we also comment on a geometric
quantum blowing-up statement, originating in the work of Osborne and
Winter~\cite{OW}, and make the connection with classical blowing up explicit.

Explicitly, writing $C(\N)$ for the classical capacity, and $p_{{\rm succ},n}$ for the average decoding success probability for uniformly distributed messages, we prove the following in Section~\ref{sec:classical-capacity}. 

\begin{theorem}[Exponential strong converse for classical capacity]
\label{thm:classical-exponential}
For every finite-dimensional memoryless channel $\N$ and every
$\gamma>0$, there are $\alpha_\gamma>0$ and $n_0$ such that every
$(n,M_n)$ classical code with $n\ge n_0$ and
$\log M_n/n\ge C(\N)+\gamma$ has average success probability
$p_{{\rm succ},n}<2^{-\alpha_\gamma n}$. 
\end{theorem}

Section~\ref{sec:private-capacity} treats private classical communication
at the regularized private-information capacity~\cite{Devetak}.
Using the privacy-test formulation of Wilde, Tomamichel, and
Berta~\cite{WTB}, we control smooth secrecy and reliability entropies.
One-shot private achievability~\cite{RenesRenner} converts these bounds
into a quantum blowing-up lemma for private communication. The resulting
secret-key fidelity decays exponentially above capacity; if the secrecy
error tends to zero, the decoding success probability also tends to zero.

Writing $P(\N)$ for the private capacity, we obtain the following result. The secret-key fidelity $F_{{\rm key},n}$ measures fidelity to a uniform shared key independent of the environment, as defined in~\eqref{eq:private-fidelity}. Privacy-test projectors and their acceptance probabilities are defined in Section~\ref{sec:choose-test}.

\begin{theorem}[Exponential strong converse for private capacity]
\label{thm:private-exponential}
For every finite-dimensional memoryless channel $\N$ and every
$\gamma>0$, there are $\alpha_\gamma>0$ and $n_0$ such that every
$(n,M_n)$ private code with $n\ge n_0$ and
$\log M_n/n\ge P(\N)+\gamma$ has
$F_{{\rm key},n}<2^{-\alpha_\gamma n}$.
More generally, every $M_n$-key privacy test on its coherified output
has acceptance $f_n<2^{-\alpha_\gamma n}$.
\end{theorem}

Concurrent work by Cheng and Tomamichel~\cite{ChengTomamichel} establishes exponential strong converses for quantum and classical communication over arbitrary finite-dimensional memoryless channels using an Arimoto approach based on R\'enyi information measures, orthogonal to our method.

\section{Preliminaries}\label{sec:setup}

All Hilbert spaces considered here are finite dimensional, and logarithms are in base $2$. To simplify the expressions we assume that operators act as identities on omitted subsystems, and tensor factors are reordered when necessary. For states $\rho,\sigma$, we use the squared fidelity $F(\rho,\sigma)=\norm{\sqrt\rho\sqrt\sigma}_1^2.$
For a unit vector $\ket\phi$, this reduces to $F(\ketbra{\phi}{\phi},\sigma)=\bra\phi\sigma\ket\phi$. Let $h(t)=-t\log t -(1-t)\log(1-t)$ denote the binary entropy function.

Fix a quantum channel $\N:A\to B$ with Stinespring isometry $U:A\to B\otimes E$, where $E$ is the environment, so $U^\dagger U=\I_A$ and 
\begin{align}
\N(X)=\Tr_E(UXU^\dagger).
\end{align}
An entanglement-transmission code of dimension $M$ consists of an encoding channel $\Enc:L\to A$ and a decoding channel $\Dec:B\to\widehat L$, where $L, \widehat L$ have dimension $M$. The encoder and decoder are arbitrary completely positive trace-preserving maps; in particular, the encoded state may be mixed. With respect to fixed orthonormal bases, write $\ket\Phi_{RL}=M^{-1/2}\sum_{j=1}^M\ket j_R\otimes\ket j_L$ where $R$ is an $M$-dimensional reference register, and define $\ket\Phi_{R\widehat L}$ in the same way. The code's entanglement fidelity (Figure~\ref{fig:code}) is
\begin{equation}\label{eq:codeF}
 F=F\!\left(\ketbra{\Phi}{\Phi}_{R\widehat L},\,
 (\id_R\otimes\Dec\circ\N\circ\Enc)(\ketbra{\Phi}{\Phi}_{RL})\right).
\end{equation}

\begin{figure}[t]
\centering
\begin{tikzpicture}[x=1cm,y=0.8cm,>=stealth,thick]
 \draw[->] (0,1.6)--(10,1.6);
 \node[above] at (0.55,1.6) {$R$};
 \draw[->] (0,0)--(1.25,0);
 \node[above] at (0.55,0) {$L$};
 \draw (1.25,-0.45) rectangle (2.45,0.45);
 \node at (1.85,0) {$\Enc$};
 \draw[->] (2.45,0)--(4,0);
 \node[above] at (3.2,0) {$A$};
 \draw (4,-1.45) rectangle (5.3,0.45);
 \node at (4.65,-0.5) {$U$};
 \draw[->] (5.3,0)--(7,0);
 \node[above] at (6.15,0) {$B$};
 \draw (7,-0.45) rectangle (8.2,0.45);
 \node at (7.6,0) {$\Dec$};
 \draw[->] (8.2,0)--(10,0);
 \node[above] at (9.1,0) {$\widehat L$};
 \draw[->] (5.3,-1.05)--(10,-1.05);
 \node[above] at (9.1,-1.05) {$E$};
 \draw[decorate,decoration={brace,mirror,amplitude=5pt}]
 (-0.15,1.6)--(-0.15,0);
 \node[left] at (-0.35,0.8) {$\ket\Phi_{RL}$};
 \draw[decorate,decoration={brace,amplitude=5pt}]
 (10.15,1.6)--(10.15,0);
 \node[right,align=center] at (10.4,0.8)
 {$\ketbra{\Phi}{\Phi}_{R\widehat L}$\\[-1pt]\scriptsize fidelity test};
\end{tikzpicture}
\caption{Entanglement transmission through $\N$, represented by its Stinespring isometry $U$. The encoding and decoding channels are denoted by $\Enc$ and $\Dec$ respectively. The overlap with the maximally entangled target gives the fidelity in \eqref{eq:codeF}.}
\label{fig:code}
\end{figure}
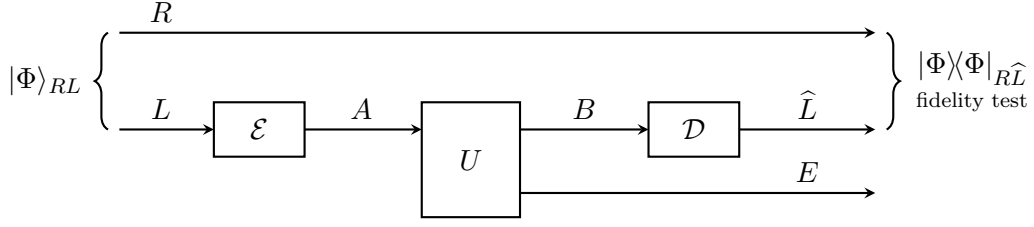

The fidelity formula also applies to sub-normalized states when at least one of the states is normalized. For these states, the purified distance is $P(\rho,\sigma)=\sqrt{1-F(\rho,\sigma)}$.
For a sub-normalized state $\rho_{ER}$, we use the conditional min-entropy
\begin{equation}
 \Hmin(R|E)_\rho
 =-\log\inf\{\Tr\sigma_E:\sigma_E\ge0,\ \rho_{ER}\le\sigma_E\otimes\I_R\}.
\end{equation}
For a normalized state $\rho_{ER}$ and $0\le\delta<1$, the smooth min-entropy is~\cite[Definition~12]{TCR}
\begin{equation}\label{eq:smoothing}
 H_{\min}^{\delta}(R|E)_\rho
 =\sup\bigl\{\Hmin(R|E)_\omega:\omega_{ER}\ge0,\ \Tr\omega_{ER}\le1,\ P(\omega,\rho)\le\delta\bigr\},
\end{equation}
where $\omega_{ER}$ ranges over operators on $E\otimes R$.

For a state $\omega_{XB}$, the conditional max-entropy can be defined
by duality, $H_{\max}(X|B)_\omega=-H_{\min}(X|C)_\omega$, where
$\omega_{XBC}$ is any purification~\cite{TCR}. Its smooth version is
\begin{align}
 H_{\max}^{\eta}(X|B)_\omega
 =\inf_{\substack{\sigma_{XB}\ge0,\ \Tr\sigma_{XB}\le1\\
                   P(\sigma_{XB},\omega_{XB})\le\eta}}
       H_{\max}(X|B)_\sigma.
\end{align}
For a classical register with distribution $q$, $H_{\min}(X)=-\log\max_xq_x$.

\begin{fact}[One-shot quantum achievability {\cite[Proposition~20]{MW}; \cite[Eqs.~(32)--(33)]{TBR}}]\label{fact:oneshot}
Let $\ket\xi_{RA}$ be any unit vector and let $\ket\psi_{RBE}=(\I_R\otimes U)\ket\xi_{RA}$. For $0<\eta<\delta<1$, there exists an entanglement-transmission code with fidelity $F$ and dimension $M$ satisfying\footnote{The restriction $\eta\in(0,\epsilon]$ accompanying
Eqs.~(32)--(33) of \cite{TBR} should read
$\eta\in(0,\sqrt{\epsilon})$, as follows from
\cite[Proposition~20]{MW}, where the smoothing parameter and the
additional error parameter can be chosen independently and their
sum bounds the purified distance. Here $\epsilon=\delta^2$,
so the corresponding range is $0<\eta<\delta$.}
\begin{equation}\label{eq:oneshot}
 F\ge1-\delta^2,\qquad
 \log M\ge H_{\min}^{\delta-\eta}(R|E)_{\psi_{ER}}-4\log\frac1\eta-1.
\end{equation}
\end{fact}

We record two partial-trace bounds and a bound for convex combinations of vectors for later use.

\begin{lemma}[Marginal bounds]\label{lem:test-marginal}
\textup{(i)} Let $S_{BE}\ge0$ and let $K:B\to B'$ satisfy
$K^\dagger K\le\I_B$. Then
\begin{equation}\label{eq:receiver-contraction}
 \Tr_{B'}\!\left[(K\otimes\I_E)S_{BE}(K^\dagger\otimes\I_E)\right]
 \le\Tr_B S_{BE}.
\end{equation}
\textup{(ii)} Let $Q_{BR}\ge0$ and $\rho_{BER}\ge0$, with no normalization assumed. Set
\begin{align}
 \tau_E=\Tr_{BR}(Q^{1/2}\rho Q^{1/2}),\qquad
 c=\op{\Tr_RQ}.
\end{align}
Then
\begin{equation}\label{eq:test-marginal}
 \Tr_B(Q\rho Q)\le c\,\tau_E\otimes\I_R.
\end{equation}

\textup{(iii)} For any vectors $\ket{v_1},\ldots,\ket{v_m}$ and any probability distribution $p_1, \dots, p_m$ we have
\begin{equation}\label{eq:coherent-sum}
 \left(\sum_{j=1}^m p_j \ket{v_j}\right)
 \left(\sum_{j=1}^m p_j\bra{v_j}\right)
 \leq  \sum_{j=1}^m p_j  \ketbra{v_j}{v_j}.
\end{equation}

\end{lemma}
\begin{proof}
For \textup{(i)}, test the difference between the right- and left-hand
sides against any $Z_E\ge0$. By cyclicity of the trace,~\eqref{eq:receiver-contraction} is equivalent to
\begin{align}
 \Tr\!\left[\bigl((\I_B-K^\dagger K)\otimes Z_E\bigr)S_{BE}\right]\ge0,
\end{align}
which holds since both operators in the product are positive. 

For \textup{(ii)}, by linearity in $\rho$, it suffices to prove the claim for
$\rho=\ketbra{\psi}{\psi}$. Define 
\begin{align}
\ket{\chi} = (\tau_E^{-1/2}\otimes \I_{BR}) Q^{1/2}\ket\psi,
\end{align}
where the inverse is taken on the support of \(\tau_E\). Then $\chi_E = \Tr_{BR} \ketbra\chi\chi = \tau_E^{-1/2}\tau_E  \tau_E^{-1/2}\leq \I_E.$
Since the two marginals of a pure operator have the same nonzero eigenvalues, we also have $\chi_{BR}=\Tr_E \ketbra\chi\chi\le \I_{BR}$.
We put \(|v\rangle=Q^{1/2}|\chi\rangle\). Consequently,
\begin{align} 
v_B= \Tr_{RE} \ketbra vv =\operatorname{tr}_R\!\left(Q^{1/2}\chi_{BR}Q^{1/2}\right) \le \operatorname{tr}_RQ \le c\I_B. 
\end{align}
Applying the same observation about pure-state marginals gives $v_{ER} = \Tr_B \ketbra vv\le c\I_{ER}$. Finally,
$$ Q|\psi\rangle=(\tau_E^{1/2}\otimes \I_{BR})|v\rangle, $$
and therefore
\begin{align}
\Tr_B(Q\ketbra\psi\psi Q) = \big(\tau_E^{1/2}\otimes \I_R\big) v_{ER} \big(\tau_E^{1/2}\otimes \I_R\big)\leq c\tau_E\otimes \I_R.
\end{align}

For \textup{(iii)} it suffices to show that for any vector $\ket w$,
\begin{equation}
 \left|\sum_{j=1}^m p_j \langle w|v_j\rangle\right|^2
 \le \sum_{j=1}^mp_j|\langle w|v_j\rangle|^2,
\end{equation}
which holds by the Cauchy--Schwarz inequality.
\end{proof}

Our quantum blowing-up lemma is expressed in terms of the projective tensor norm. For an operator $Z$ on $B\otimes E$, the projective tensor norm induced by the local operator norms is
\begin{equation}\label{eq:projective}
 \pn{Z}=\inf \sum_j\op{X_j}\op{Y_j}.
\end{equation}
Here the infimum is taken over all expressions $Z=\sum_jX_j\otimes Y_j$
where $X_j$ acts on $B$ and $Y_j$ acts on $E$. In finite dimensions, a compactness argument shows that this infimum is attained. 
The projective tensor norm satisfies the triangle inequality and is the largest reasonable cross norm
on the tensor product of the operator spaces equipped with the operator norm~\cite{Ryan2002}.

\begin{lemma}[Dimension bound for the projective tensor norm]\label{lem:projective-dimension}
For every operator $Z$ on $B\otimes E$,
\begin{equation}\label{eq:projective-dimension}
 \pn{Z}\le\min\big\{(\dim B)^2,(\dim E)^2\big\}\,\op{Z}.
\end{equation}
\end{lemma}
\begin{proof}
Fix an orthonormal basis $\{\ket i:\, i=1, \dots, \dim B\}$ of $B$, and expand
\begin{align}
 Z=\sum_{i,j=1}^{\dim B}\ketbra{i}{j}\otimes Z_{ij},\qquad
 Z_{ij}=(\bra i\otimes\I_E)Z(\ket j\otimes\I_E).
\end{align}
We note that $\|(\bra i\otimes\I_E)\|_{\infty} =\|(\ket j\otimes\I_E)\|_{\infty}=1$, so
$\op{Z_{ij}}\le\op Z$. We also have $\op{\ketbra{i}{j}}=1$, so the definition
of the projective tensor norm gives
\begin{align}
 \pn Z\le\sum_{i,j=1}^{\dim B}\op{Z_{ij}}\le (\dim B)^2\op Z.
\end{align}
We similarly have 
$\pn Z\le(\dim E)^2\op Z$, which proves \eqref{eq:projective-dimension}.
\end{proof}

\section{Fully quantum blowing-up lemma}\label{sec:blowingup}
\begin{theorem}[Quantum blowing-up lemma]\label{thm:upgrade} 
Consider an $M$-dimensional code $(\Enc,\Dec)$ with entanglement fidelity $F>0$, and fix $0<\varepsilon<1$. Let $\Pi=UU^\dagger$ be the orthogonal projector onto the image of the channel's Stinespring isometry $U:A\to B\otimes E$. Let ${\widetilde\Pi}$ be an operator approximation of $\Pi$ satisfying, for some $\Gamma>0$,
\begin{equation}\label{eq:threshold}
 {\widetilde\Pi}\,\Pi=\Pi\,{\widetilde\Pi}=\Pi,\qquad
 \pn {\widetilde\Pi}\le\Gamma,\qquad
 \op{{\widetilde\Pi}-\Pi}\le\frac{\varepsilon\sqrt F}{2}.
\end{equation}
Then there is a code with entanglement fidelity at least $1-\varepsilon^2$ and dimension $M'$ satisfying
\begin{equation}\label{eq:upgrade}
 \boxed{\quad\log M'\ge\log M+\log F-2\log\Gamma-4\log\frac2\varepsilon-1.\quad}
\end{equation}
\end{theorem}

\begin{proof}
\textit{Step 1: pull the fidelity test back through the decoder.}
Pull the maximally entangled projector back through the decoder, and let $\rho_{BER}$ be the encoded state after the Stinespring isometry:
\begin{align}\label{eq:codetest}
 Q_{BR}&=(\Dec^\dagger\otimes\id_R)(\ketbra{\Phi}{\Phi}_{\widehat L R}),\\
 \rho_{BER}&= (U\otimes \I_R)\big((\Enc\otimes \id_R)(\ketbra{\Phi}{\Phi}_{LR})\big) (U^\dagger\otimes\I_R).
\end{align}
The adjoint $\Dec^\dagger$ is completely positive and unital. Hence
\begin{equation}\label{eq:codetestprops}
 0\le Q\le\I_{BR},\qquad
 \Tr_RQ=\Dec^\dagger(\I_{\widehat L}/M)=\I_B/M,\qquad
 \Pi\rho\Pi=\rho,\qquad\Tr(\rho Q)=F.
\end{equation}

\medskip

\noindent
\textit{Step 2: choose a vector in the channel image.}
Let $\lambda=\op{\Pi Q\Pi}$. Since $\rho$ is obtained by applying the channel isometry $U$ we have $\Pi\rho\Pi=\rho$ and
\begin{equation}\label{eq:lambdaF}
0< F= \Tr(\rho Q) =\Tr(\Pi \rho \Pi Q)=\Tr(\rho\Pi Q\Pi)\le\lambda.
\end{equation}
Let $\ket\psi$ be a normalized eigenvector of $\Pi Q\Pi$ with maximal eigenvalue $\lambda$. Since $\lambda>0$, we have $\Pi\ket\psi=\ket\psi$ and $\Pi Q\ket\psi=\lambda\ket\psi$. Set
\begin{equation}\label{eq:smoothedvector}
 \ket{\widetilde\psi}=\frac{{\widetilde\Pi}Q\ket\psi}{\norm{{\widetilde\Pi}Q\ket\psi}}.
\end{equation}

\medskip

\noindent
\textit{Step 3: control the distance to the smoothed vector.}
Since $Q^2\le Q$, we have $\norm{Q\ket\psi}^2 = \bra\psi Q^2\ket \psi \leq \bra \psi Q\ket \psi = \bra \psi \Pi Q\Pi\ket \psi = \lambda$. The identities ${\widetilde\Pi}\,\Pi=\Pi\,{\widetilde\Pi}=\Pi$ give $\Pi(\widetilde \Pi-\Pi)=0$, implying that 
\begin{equation}\label{eq:orthog}
 \norm{{\widetilde\Pi}Q\ket\psi}^2=\norm{\Pi Q\ket\psi}^2+\norm{({\widetilde\Pi}-\Pi)Q\ket\psi}^2=\lambda^2+\norm{({\widetilde\Pi}-\Pi)Q\ket\psi}^2
 \le\lambda^2+\lambda\op{{\widetilde\Pi}-\Pi}^2.
\end{equation}
In particular, $\norm{{\widetilde\Pi}Q\ket\psi}^2\ge\lambda^2>0$, so $\ket{\widetilde\psi}$ is well defined. Moreover, we have
\begin{align}
\braket{ \psi}{\widetilde\psi} = \frac{\bra \psi \widetilde \Pi Q\ket\psi}{\norm{{\widetilde\Pi}Q\ket\psi}} = \frac{\bra \psi  \Pi \widetilde \Pi Q\ket\psi}{\norm{{\widetilde\Pi}Q\ket\psi}} = \frac{\bra \psi  \Pi  Q\ket\psi}{\norm{{\widetilde\Pi}Q\ket\psi}} = \frac{\lambda}{\norm{{\widetilde\Pi}Q\ket\psi}},
\end{align}
 and
\begin{align} \label{eq:phiP}
 P(\ket{\widetilde\psi},\ket\psi)^2 =1-\frac{\lambda^2}{\norm{{\widetilde\Pi}Q\ket\psi}^2}
 \le 1-\frac{\lambda^2}{\lambda^2+\lambda\op{{\widetilde\Pi}-\Pi}^2}  =\frac{\op{{\widetilde\Pi}-\Pi}^2}{\lambda+\op{{\widetilde\Pi}-\Pi}^2}\le\frac{\varepsilon^2}{4},
\end{align}
where the last inequality uses \eqref{eq:threshold} and $F\le\lambda$.

\medskip

\noindent
\textit{Step 4: dominate the complementary marginal.}
Set
\begin{equation}\label{eq:tau}
 \tau_E=\Tr_{BR}\!\left(Q^{1/2}\ketbra{\psi}{\psi} Q^{1/2}\right),
 \qquad\Tr\tau_E=\lambda.
\end{equation}
Applying Lemma~\ref{lem:test-marginal}\textup{(ii)} to $\ketbra{\psi}{\psi}$ and using
$\Tr_RQ=\I_B/M$ yield
\begin{equation}\label{eq:Qdom}
 \Tr_B(Q\ketbra{\psi}{\psi} Q)\le\frac1M\,\tau_E\otimes\I_R.
\end{equation}

By the definition of the projective tensor norm and compactness in finite dimensions, we can write
\begin{equation}\label{eq:Kdecomposition}
 {\widetilde\Pi}=\Gamma\sum_i \beta_iX_i\otimes Y_i,\qquad
 \beta_i\ge0,\quad\sum_i\beta_i=1,\quad\op{X_i},\op{Y_i}\le1,
\end{equation}
where $X_i$ acts on $B$ and $Y_i$ acts on $E$. Apply Lemma~\ref{lem:test-marginal}\textup{(iii)} with
$\ket{v_i}=(X_i\otimes Y_i\otimes\I_R)Q\ket\psi$, and then use
Lemma~\ref{lem:test-marginal}\textup{(i)} with $K=X_i$ to remove the contraction $X_i$ under the partial trace to obtain
 \begin{align}
 \widetilde\psi_{ER} = \Tr_B\ketbra{\widetilde\psi}{\widetilde\psi} &= \frac{\Gamma^2}{\norm{{\widetilde\Pi}Q\ket\psi}^2}
\Tr_{B}\Bigg(\Big(  \sum_i \beta_i\ket{v_i}\Big)\Big(  \sum_i \beta_i\bra{v_i}\Big)\Bigg)\\
 &\le\frac{\Gamma^2}{\norm{{\widetilde\Pi}Q\ket\psi}^2}
 \sum_i\beta_i(Y_i\otimes\I_R)\Tr_B(Q\ketbra{\psi}{\psi} Q)(Y_i^\dagger\otimes\I_R)\\
 &\le\frac{\Gamma^2}{M\norm{{\widetilde\Pi}Q\ket\psi}^2}
 \Big(\sum_i\beta_iY_i\tau_EY_i^\dagger\Big)\otimes\I_R, \label{eq:phidom}
 \end{align}
where the last line follows from~\eqref{eq:Qdom}.
The operator in parentheses has trace at most $\lambda$, because $\sum_i\beta_i=1$, $\Tr\tau_E=\lambda$, and $Y_i^\dagger Y_i\le\I_E$. Thus
\begin{equation}
 \Hmin(R|E)_{\widetilde\psi}\ge\log\frac{M\norm{{\widetilde\Pi}Q\ket\psi}^2}{\Gamma^2\lambda}\geq \log\frac{M\lambda^2}{\Gamma^2\lambda}=\log M+\log\lambda-2\log\Gamma.
\end{equation}
\medskip
\noindent
\textit{Step 5: apply one-shot quantum achievability.}
Monotonicity of fidelity under partial trace and \eqref{eq:phiP} give $P(\widetilde\psi_{ER},\psi_{ER})\le\varepsilon/2$. Hence $\widetilde\psi_{ER}$ is allowed in the smoothing in \eqref{eq:smoothing}, and
\begin{equation}\label{eq:entropyupgrade}
 H_{\min}^{\varepsilon/2}(R|E)_{\psi}
 \ge\log M+\log\lambda-2\log\Gamma\geq \log M+\log F-2\log\Gamma.
\end{equation}
Because $\Pi\ket\psi=\ket\psi$, the input vector
\begin{equation}\label{eq:inputpullback}
 \ket\xi_{AR}=(\I_R\otimes U^\dagger)\ket\psi
\end{equation}
is normalized and produces $\ket\psi$ after applying the channel isometry $U$. Thus $\ket{\xi}_{AR}$ can be used in Fact~\ref{fact:oneshot} with $\delta=\varepsilon$ and $\eta=\varepsilon/2$ to establish~\eqref{eq:upgrade}. If the stated lower bound is negative, the dimension-one code suffices.
\end{proof}

\section{Approximating the product image projector}\label{sec:approximation}

We now construct the approximating operator $\widetilde \Pi$ required by the quantum blowing-up lemma for $n$ uses of $\N$.  The Stinespring image projector of $\N^{\otimes n}$ is $\Pi^{\otimes n}$, acting on $B^nE^n$. More generally, the same construction approximates a product of arbitrary local projections on copies of $B\otimes E$.

To construct $\widetilde \Pi$, we use a low-degree polynomial approximating the NOR function on the Boolean cube~\cite{Sherstov}. This polynomial is applied to the sum of the local complementary projectors $\I-\Pi_i$ corresponding to the $i$-th use of the channel. Its degree controls both the approximation error and the projective tensor norm. 

\begin{fact}[NOR approximation {\cite[Theorem~4.5]{Sherstov}}]\label{fact:polynomial}
There is a universal constant $c>0$ such that, for any integers $n$ and $D$ satisfying $3\le D\le n$, there is a real polynomial $p_{n,D}$ of degree at most $D$ satisfying
\begin{align}\label{eq:polyinput}
 p_{n,D}(0)=1,\qquad
 \max_{k\in\{1,\ldots,n\}}|p_{n,D}(k)|
 \le2^{-cD^2/n}, \qquad \max_{k\in\{0,\ldots,n\}}|p_{n,D}(k)|\le1.
\end{align}
\end{fact}

\begin{lemma}[Polynomial projector approximation]\label{lem:polynomial}
Let $\Pi_1,\ldots,\Pi_n$ be arbitrary orthogonal projections on copies
$B_i\otimes E_i$ of the same space $B\otimes E$. Put
$\Pi=\bigotimes_{i=1}^n\Pi_i$ and
$\kappa=\min\{(\dim B)^2,(\dim E)^2\}$.
For integers $3\le D\le n$, there is an operator ${\widetilde\Pi}_D$ satisfying
\begin{equation}\label{eq:polyK}
 {\widetilde\Pi}_D\Pi=\Pi{\widetilde\Pi}_D=\Pi,\qquad
 \op{{\widetilde\Pi}_D-\Pi}\le2^{-cD^2/n}.
\end{equation}
Moreover, ${\widetilde\Pi}_D$ is a sum of operators supported on at most
$D$ pairs $B_iE_i$. If $D\le n/2$, it also satisfies
\begin{equation}\label{eq:Gamma}
 \pn{{\widetilde\Pi}_D}\le\Gamma_{n,D}:=(2\kappa)^D2^{n h(D/n)}.
\end{equation}
Here $h(\cdot)$ is the binary entropy function, and the projective tensor norm is based on the bipartition $B^n:E^n$.
\end{lemma}

The proof of this lemma is inspired by the low-degree polynomial construction of Kondra et al.~\cite[Lemma~S3]{Kondra}.

\begin{proof}
Regard each $\Pi_i$ as acting on $B_iE_i$, with identities on the other factors. Set $\Delta_i=\I-\Pi_i$ and $L=\sum_i\Delta_i$. These projections commute, so $L$ has spectrum in $\{0,\ldots,n\}$ and the projector onto the kernel of $L$ is $\Pi=\bigotimes_{i=1}^n\Pi_i$. Thus, choosing ${\widetilde\Pi}_D=p_{n,D}(L)$ and applying Fact~\ref{fact:polynomial} give \eqref{eq:polyK}. We next establish the support property and, when $D\le n/2$, the bound~\eqref{eq:Gamma}. 

To expand $p_{n,D}(L)$ into terms supported on small subsets of tensor
factors, write the polynomial in the binomial basis. Newton's interpolation formula gives
\begin{equation}\label{eq:newton-coefficients}
 p_{n,D}(x)=\sum_{s=0}^D a_s\binom xs,\qquad
 a_s=\sum_{k=0}^s(-1)^{s-k}\binom sk p_{n,D}(k).
\end{equation}
The coefficient $a_s$ is the $s$-th forward difference of $p_{n,D}(x)$ at zero.
Indeed, taking a forward difference, $f(x)\mapsto f(x+1)-f(x)$,
sends $\binom{x}{s}$ to $\binom{x}{s-1}$; taking $s$ differences
and evaluating at zero therefore isolates $a_s$. Expanding these successive
differences gives the formula for $a_s$ in \eqref{eq:newton-coefficients}.
Since $s\le D\le n$, Fact~\ref{fact:polynomial} gives 
$|p_{n,D}(k)|\leq 1$ for every integer $0\leq k\leq n$. Consequently,
\begin{equation}\label{eq:newton-coefficient-bound}
 |a_s|\le\sum_{k=0}^s\binom sk |p_{n,D}(k)|
 \le\sum_{k=0}^s\binom sk=2^s.
\end{equation}

The operators $\Delta_i$ commute and are projections, so they admit a
common eigenbasis with eigenvalues in $\{0,1\}$. On a common eigenvector
with exactly $k$ eigenvalues equal to one, $L$ acts as $k$, while
$\prod_{i\in S}\Delta_i$ acts as $1$ precisely when every index in $S$
is among these $k$ active indices. There are $\binom{k}{s}$ such subsets
of size $s$. Thus
$\binom{L}{s}=\sum_{|S|=s}\prod_{i\in S}\Delta_i$, and
\eqref{eq:newton-coefficients} yields
\begin{equation}\label{eq:newtonK}
 {\widetilde\Pi}_D=p_{n,D}(L)=\sum_{s=0}^D a_s\binom Ls=\sum_{s=0}^D a_s\sum_{\substack{S\subseteq\{1,\ldots,n\}\\|S|=s}}
 \prod_{i\in S}\Delta_i.
\end{equation}
The empty product, corresponding to $s=0$, is the identity.
Each term acts on at most $D$ pairs, proving the support assertion.
For the projective-norm bound, suppose now that $D\le n/2$.

For $|S|=s$, the operator $\prod_{i\in S}\Delta_i$ is a projection on
$B_SE_S$, tensored with identities on the remaining registers. Its operator
norm is at most one. Applying Lemma~\ref{lem:projective-dimension} on
$B_S:E_S$, whose local dimensions are $(\dim B)^s$ and $(\dim E)^s$,
gives a projective tensor norm of at most $\kappa^s$.
Adjoining the local identities on $B_{S^c}$ and $E_{S^c}$ does not increase
this bound, since these identities have operator norm one.
The triangle inequality, the $\binom ns$ subsets of size $s$, and
\eqref{eq:newton-coefficient-bound} now give
\begin{equation}\label{eq:projective-sum-bound}
 \pn{{\widetilde\Pi}_D}
 \le\sum_{s=0}^D |a_s|\binom ns\kappa^s
 \le\sum_{s=0}^D\binom ns(2\kappa)^s
 \le(2\kappa)^D\sum_{s=0}^D\binom ns,
\end{equation}
where the last inequality uses $2\kappa\ge1$ and $s\le D$.
Put $t=D/n\le1/2$. Then $t^s(1-t)^{n-s}\ge t^D(1-t)^{n-D}=2^{-nh(t)}$ for $s\le D$. The binomial theorem yields
\begin{equation}
 1\ge\sum_{s=0}^D\binom ns t^s(1-t)^{n-s}
 \ge2^{-nh(t)}\sum_{s=0}^D\binom ns.
\end{equation}
Combining the last two bounds proves \eqref{eq:Gamma}.
\end{proof}

\section{Strong converse for quantum capacity}\label{sec:strongconverse}
We give the proof of Theorem~\ref{thm:strongconverse}, using the
one-shot conversion and polynomial approximation from
Sections~\ref{sec:blowingup} and~\ref{sec:approximation}.
For $n$ uses of $\N$, an entanglement-transmission code consists of channels $\Enc_n:L_n\to A^n$ and $\Dec_n:B^n\to\widehat L_n$, where $\dim L_n=\dim\widehat L_n=\dim R_n=M_n$. The reference $R_n$ is untouched. Its entanglement fidelity $F_n$ is defined by \eqref{eq:codeF}, with $\N$, $\Enc$, and $\Dec$ replaced by $\N^{\otimes n}$, $\Enc_n$, and $\Dec_n$, and with maximally entangled input and target states on $R_nL_n$ and $R_n\widehat L_n$, respectively. Both maps may act jointly across all channel uses; in particular, the encoded inputs may be entangled across the $n$ copies. The rate is $(\log M_n)/n$ qubits per channel use.

The quantum capacity $Q(\N)$ is the supremum of rates $r$ for which a sequence of these codes has $\liminf_{n\to\infty}(\log M_n)/n\ge r$ and $F_n\to1$. We use the following weak converse bound,
which follows from the characterization of quantum capacity in terms of regularized coherent information.

\begin{fact}[Weak converse for quantum communication {\cite[Theorem~5 (converse) and Proposition~7]{Devetak}}]\label{fact:weakconverse}
For every $r_0>Q(\N)$, there are $0<\varepsilon<1$ and an integer $n_0$ such that, for every $n\ge n_0$, any code for $\N^{\otimes n}$ with fidelity at least $1-\varepsilon^2$ has dimension $M'$ satisfying $\log M'\le nr_0$.
\end{fact}

\begin{proof}[Proof of Theorem~\ref{thm:strongconverse}]
The code conversion in Theorem~\ref{thm:upgrade} also applies to the arbitrary
encoded states allowed in Theorem~\ref{thm:strongconverse}. Indeed, its proof uses the initial state only to obtain a normalized
output $\rho$ supported on the Stinespring image with $\Tr(\rho Q)=F$.
For an arbitrary state on $RA^n$ and an $M$-dimensional decoder target, these
properties and $\Tr_RQ=\I_{B^n}/M$ still hold. The same proof therefore produces
an entanglement-transmission code, to which Fact~\ref{fact:weakconverse} applies.

Set $r=Q(\N)+\gamma$. Choose $Q(\N)<r_0<r$, and fix $\varepsilon$ and $n_0$ from Fact~\ref{fact:weakconverse}. Fix a Stinespring isometry $U$ of $\N$ with image projector $\Pi = UU^\dagger$, and set $\kappa=\min\{(\dim B)^2,(\dim E)^2\}$. Choose $0<\theta<1/2$ small enough that $2h(\theta)+2\theta\log(2\kappa)<r-r_0$, and set
\begin{equation}\label{eq:exponentchoice}
 \alpha_\gamma=\min\left\{\frac{c\theta^2}{8},\ \frac{r-r_0}4-\frac12h(\theta)-\frac\theta2\log(2\kappa)\right\}>0,
\end{equation}
where $c$ is the constant from Fact~\ref{fact:polynomial}. All these choices are independent of $n$ and of the code.

For sufficiently large $n$, put $D=\lfloor\theta n\rfloor$. Then $3\le D\le n/2$ and $D\ge\theta n/2$. Lemma~\ref{lem:polynomial} gives an operator ${\widetilde\Pi}_D$ satisfying ${\widetilde\Pi}_D\Pi^{\otimes n}=\Pi^{\otimes n}{\widetilde\Pi}_D=\Pi^{\otimes n}$, with
\begin{align}\label{eq:scapprox}
& \op{{\widetilde\Pi}_D-\Pi^{\otimes n}}\le 2^{-cD^2/n}
 \le2^{-nc\theta^2/4}\le2^{-2\alpha_\gamma n},\\
& \log\Gamma_{n,D}=n h(D/n)+D\log(2\kappa)
 \le n\bigl[h(\theta)+\theta\log(2\kappa)\bigr].
\end{align}
The last inequality uses the monotonicity of $h(\cdot)$  on $(0,1/2)$.

Suppose a code of dimension $M\ge2^{nr}$ had fidelity $F>2^{-2\alpha_\gamma n}$. Increasing $n_0$ so that $2^{-\alpha_\gamma n}\le\varepsilon/2$ for all $n\ge n_0$, we obtain
\begin{equation}
 \op{{\widetilde\Pi}_D-\Pi^{\otimes n}}\le2^{-2\alpha_\gamma n}
 \le\frac\varepsilon2\,2^{-\alpha_\gamma n}<\frac{\varepsilon\sqrt F}{2}.
\end{equation}
Theorem~\ref{thm:upgrade}, applied to the channel $\N^{\otimes n}$, therefore produces a code with fidelity at least $1-\varepsilon^2$ and dimension $M'$ satisfying
\begin{align}
 \log M'&\ge\log M+\log F-2\log\Gamma_{n,D}-4\log\frac2\varepsilon-1\\
 &>n\bigl[r-2\alpha_\gamma-2h(\theta)-2\theta\log(2\kappa)\bigr]-4\log\frac2\varepsilon-1.
\end{align}
By \eqref{eq:exponentchoice}, $r-2\alpha_\gamma-2h(\theta)-2\theta\log(2\kappa)\ge r_0+2\alpha_\gamma$. Increasing $n_0$ once more so that $2n\alpha_\gamma>4\log(2/\varepsilon)+1$, we conclude $\log M'>nr_0$, contradicting Fact~\ref{fact:weakconverse}. Hence every such code satisfies $F\le2^{-2\alpha_\gamma n}$. Since $\alpha_\gamma>0$, this gives $F_n<2^{-\alpha_\gamma n}$ for all sufficiently large $n$, as claimed.
\end{proof}

\section{Strong converse for classical capacity}\label{sec:classical-capacity}
We now apply the channel-image approximation to classical communication and prove an exponential strong converse bound for the classical capacity of all quantum channels.

The strong converse for classical--quantum (c-q) channels has been established in~\cite{OgawaNagaoka,WinterCoding}. In particular, it can be recovered from the following c-q blowing-up lemma. 

Let $B_1,\ldots,B_n$ be copies of the same finite-dimensional space, and let $P$ be an orthogonal projection acting on $B_1,\ldots,B_n$.
For $S\subseteq[n]$, write $B_S=\bigotimes_{i\in S}B_i$.
For an integer $0\le D\le n$, define $P^{[D]}$ to be the orthogonal projection onto the span of
\begin{align}
\bigcup_{X_S: |S|\leq D} \ran\big((X_S\otimes\I_{S^c}) P\big),
\end{align}
where $X_S$ runs over all operators acting on $B_S$ with $|S|\leq D$.
Thus $P^{[0]}=P$, and $P^{[D]}$ permits arbitrary changes on at most $D$ tensor factors, followed by taking the linear span. 

\begin{theorem}[c-q blowing-up lemma]\label{thm:geometric}
There is a universal constant $c>0$ with the following property.
Let $\rho=\rho_1\otimes\cdots\otimes\rho_n$ be any product state and let
$p=\Tr(\rho P)>0$. Then, for every integer $0\le D\le n$,
\begin{equation}\label{eq:starting}
 \boxed{\quad
  \Tr \left[\rho\big(\I-P^{[D]}\big)\right]
  \le \frac{1}{p}\,2^{-cD^2/n}.
 \quad}
\end{equation}
\end{theorem}

We note that when $\rho_i$'s are all diagonal, and $P$ is diagonal in the corresponding product basis, this theorem reduces to a classical blowing-up inequality. Indeed, this theorem is more
reminiscent of the classical blowing-up lemma: a subspace of nonzero
probability acquires probability close to one after enlargement on a
small number of tensor factors.

This theorem first appeared in Osborne's 2009 blog post, reporting
joint work with Winter~\cite{OW}. We present a different proof in
Appendix~\ref{app:product-blowing}, using the same ideas as our construction in Lemma~\ref{lem:polynomial}. This shows the broad applicability of our blowing-up argument.

\subsection{Quantum blowing-up lemma for classical communication}
We now turn to classical channel codes for all quantum channels $\mathcal N$. 
As before, we denote by $\Pi=UU^\dagger$ the Stinespring image projector of $\N$.
An $M$-message code for $\N$ consists of input states $\rho_1,\ldots,\rho_M$ and a decoding POVM $\{\Lambda_m\}_{m=1}^M$ on $B$.
For uniformly distributed messages, the average success probability is
\begin{equation}\label{eq:success}
 p=\frac1M\sum_{m=1}^M\Tr\bigl[\Lambda_m\N(\rho_m)\bigr].
\end{equation}
The following achievability result plays the role of Fact~\ref{fact:oneshot} for classical capacity.

\begin{fact}[One-shot classical achievability {\cite[Theorem~1]{RenesRenner}}]
\label{fact:classical-oneshot}
Let $\{q_x,\rho_x\}$ be any finite ensemble of channel inputs and put
$\omega_{XB}=\sum_xq_x\ketbra xx\otimes\N(\rho_x)$.
For every $0<\delta<1$, there is a classical code with
average decoding error at most $\delta$ and message size $M'$ satisfying
\begin{equation}\label{eq:classical-achievability}
 \log M'\ge H_{\min}(X)_\omega
       -H_{\max}^{\delta/8}(X|B)_\omega
       -4\log\frac1\delta-16.
\end{equation}
\end{fact}
The cited result has the larger term $H_{\min}^{\delta/8}(X)$ and even
controls the maximal decoding error. The weaker form above suffices here.

\begin{theorem}[Quantum blowing-up lemma for classical communication]\label{thm:oneshot}
Consider an $M$-message classical code with average success probability
$p>0$, and fix $0<\delta<1$. Suppose an operator $\widetilde\Pi$ on
$B\otimes E$ satisfies, for some $\varepsilon\ge0$ and $\Gamma>0$,
\begin{equation}\label{eq:oneshot-hyp}
 \widetilde\Pi\Pi=\Pi\widetilde\Pi=\Pi,\qquad
 \op{\widetilde\Pi-\Pi}\le\varepsilon,\qquad
 \pn{\widetilde\Pi}\le\Gamma.
\end{equation}
If $\varepsilon\le\delta\sqrt p/8$, then there is a classical code for the same channel with average success
probability at least $1-\delta$ and message size $M'$ satisfying
\begin{equation}\label{eq:classical-upgrade}
 \boxed{\quad
 \log M'\ge\log M+2\log p-2\log\Gamma
                    -4\log\frac1\delta-16.
 \quad}
\end{equation}
\end{theorem}

\begin{proof}[Proof of Theorem~\ref{thm:oneshot}]
\textit{Step 1: make the decoder projective.}
Define the isometry
$J:B\to\widehat B=B\otimes\mathbb C^M$ by
$J\ket b=\sum_m\sqrt{\Lambda_m}\ket b\otimes\ket m$.
Let $Q_m=\I_B\otimes\ketbra mm$. Then
$J^\dagger Q_mJ=\Lambda_m$, and $Q_m$'s are mutually orthogonal projections. Thus, replacing $U$ by $(J\otimes\I_E)U$ and the operators $\Pi$ and $\widetilde{\Pi}$ by $(J\otimes\I_E)\Pi(J^\dagger\otimes\I_E)$ and 
$(J\otimes\I_E)\widetilde{\Pi}(J^\dagger\otimes\I_E)$ respectively, we can ensure that the decoding POVM is projective. On the other hand, application of the isometry $J$ on the channel output preserves the assumptions in~\eqref{eq:oneshot-hyp}. In particular, it does not increase the projective norm. Moreover, any decoder on the enlarged receiver can be composed with $J$, giving a decoder for the original channel.
We can therefore assume the decoding measurement $\{Q_m:\, m \}$ is projective. From now on, for simplicity, we write $B$ for the enlarged receiver.\footnote{This reduction to a projective decoding measurement can be omitted: the argument also works directly with the original POVM, using $\Lambda_m^2\le\Lambda_m$ and a POVM Cauchy--Schwarz bound in place of orthogonality. We retain the reduction because it gives a common presentation with the private-capacity argument, where the projective representation simplifies the privacy-test construction and the marginal estimates.}

\medskip

\noindent
\textit{Step 2: choose eigenvectors and reweight the messages.}
For each $m$ let $\lambda_m=\op{\Pi Q_m\Pi}$ and note that $\Tr(\N(\rho_m) Q_m)\leq\lambda_m$. Discarding zero-success messages can only strengthen the claimed bound,
so we may assume $\lambda_m>0$ for every $m$.

For each $m$, choose a normalized eigenvector $\ket{\psi_m}$ of $\Pi Q_m\Pi$ with eigenvalue $\lambda_m$. Since $\lambda_m>0$, this vector lies in the channel image. Using also $Q_m^2=Q_m$, we obtain

\begin{equation}\label{eq:eigenvectors}
 \Pi\ket{\psi_m}=\ket{\psi_m},
 \qquad \Pi Q_m\ket{\psi_m}=\lambda_m\ket{\psi_m},
 \qquad \norm{Q_m\ket{\psi_m}}^2=\lambda_m.
\end{equation}
Because $\lambda_m$ bounds the success of the original $m$th input,
\begin{equation}\label{eq:reweighting}
 Z:=\sum_m\lambda_m\ge Mp,
 \qquad q_m:=\frac{\lambda_m}{Z}.
\end{equation}
We note that $\lambda_m=\op{\Pi Q_m\Pi}$ satisfies $0<\lambda_m\le1$. Therefore, $q_m\le1/Z$ implying that
\begin{equation}\label{eq:classical-prior-min}
 H_{\min}(X)_q\ge\log Z\ge\log M+\log p.
\end{equation}
Here the min-entropy is computed with respect to the random variable $X$ taking value $m$ with probability $q_m$.

\medskip

\noindent
\textit{Step 3: construct the two vectors.}
Define a unit vector and an unnormalized vector on $XX'BE$ by
\begin{align}
 \ket\psi
 &=\sum_{m}\sqrt{q_m}\ket m_X\ket m_{X'}\ket{\psi_m}_{BE},\label{eq:omega}\\
 \ket\theta
 &=\sum_{m}\frac{\sqrt{q_m}}{\lambda_m}
             \ket m_X\ket m_{X'}Q_m\ket{\psi_m}_{BE}.\label{eq:theta}
\end{align}
Then
\begin{equation}\label{eq:theta-identities}
 \Pi\ket\theta=\ket\psi,
 \qquad \norm{\ket\theta}^2
       =\sum_{m}\frac{q_m}{\lambda_m}
       =\frac MZ\le\frac1p.
\end{equation}
Write
\begin{align}
 \ket v=(\widetilde{\Pi}-\Pi)\ket\theta,
 \qquad \widetilde{\Pi}\ket\theta=\ket\psi+\ket v,
 \qquad 
 \ket{\widetilde\psi}=\frac{1}{\norm{\widetilde{\Pi}\ket\theta}}\widetilde{\Pi}\ket\theta.
\end{align}
The identity $\Pi \widetilde{\Pi}=\Pi$ gives $\Pi\ket v=0$ and hence $\braket\psi v=0$.
Thus $\norm{\widetilde{\Pi}\ket\theta}^2=1+\norm{\ket v}^2\ge1$ and
\begin{equation}\label{eq:close}
 P(\ket\psi,\ket{\widetilde\psi})
 =\frac{\norm{\ket v}}{\sqrt{1+\norm{\ket v}^2}} \leq \norm{(\widetilde{\Pi}-\Pi)\ket\theta}
 \le\varepsilon\norm{\ket\theta}
 \le\frac{\varepsilon}{\sqrt p}.
\end{equation}

\medskip

\noindent
\textit{Step 4: dominate the complementary marginal.}
Let 
$\theta_{XX'E}=\Tr_B\ketbra\theta\theta$.
Orthogonality of $Q_m$'s and cyclicity of trace imply
\begin{equation}\label{eq:theta-domination}
 \theta_{XX'E}
 =\sum_{m}\frac{q_m}{\lambda_m^2}
     \ketbra mm_X\otimes\ketbra mm_{X'}\otimes
          \Tr_B(Q_m\ketbra{\psi_m}{\psi_m}Q_m)
 \le\I_X\otimes\theta_{X'E},
\end{equation}
where $\theta_{X'E}=\Tr_X\theta_{XX'E}$ and
$\Tr\theta_{X'E}=\norm{\ket\theta}^2\le1/p$.

Use a decomposition of $\widetilde\Pi$ as
in~\eqref{eq:Kdecomposition}. 
Applying Lemma~\ref{lem:test-marginal}\textup{(iii), then (i)}, followed by~\eqref{eq:theta-domination}, gives
\begin{align}
 \widetilde\psi_{XX'E}
 &:=\Tr_B\ketbra{\widetilde\psi}{\widetilde\psi}\\
 &\le\frac{\Gamma^2}{\norm{\widetilde{\Pi}\ket\theta}^2}\sum_i\beta_i
   (\I_{XX'}\otimes Y_i)\theta_{XX'E}
                   (\I_{XX'}\otimes Y_i^\dagger)\\
 &\le\I_X\otimes \mu_{X'E},
 \end{align}
where
\begin{align}
 \mu_{X'E}:=\frac{\Gamma^2}{\norm{\widetilde{\Pi}\ket\theta}^2}\sum_i\beta_i(\I_{X'}\otimes Y_i)\theta_{X'E}(\I_{X'}\otimes Y_i^\dagger),
 \qquad \Tr \mu_{X'E}\le\frac{\Gamma^2}{p}.\label{eq:smoothed-domination}
\end{align}
The trace estimate uses $\norm{\widetilde{\Pi}\ket\theta}\ge1$ and $ Y_i^\dagger Y_i\le\I_{E}$.
By the definition of conditional min-entropy, this gives
$H_{\min}(X|X'E)_{\widetilde\psi}\ge-2\log\Gamma-\log(1/p)$.
Since $\ket{\widetilde\psi}_{XX'BE}$ is pure, min/max-entropy duality yields
\begin{equation}\label{eq:entropy-small}
 H_{\max}(X|B)_{\widetilde\psi}
 =-H_{\min}(X|X'E)_{\widetilde\psi}
 \le2\log\Gamma+\log\frac1p.
\end{equation}

\medskip

\noindent
\textit{Step 5: apply one-shot classical achievability.}
Since $\Pi\ket{\psi_m}=\ket{\psi_m}$, the vector $U^\dagger\ket{\psi_m}$ is a normalized channel input whose image under $U$ is $\ket{\psi_m}$. Thus the marginal of~\eqref{eq:omega} on $XB$ is the genuine channel ensemble
\begin{equation}\label{eq:genuine-ensemble}
 \psi_{XB}=\sum_mq_m\ketbra mm_X\otimes
                    \N\!\left(U^\dagger\ketbra{\psi_m}{\psi_m}U\right).
\end{equation}
By contractivity of purified distance under partial trace,
\eqref{eq:close} gives $P(\psi_{XB},\widetilde\psi_{XB})\le\frac{\varepsilon}{\sqrt p}$.
Thus~\eqref{eq:entropy-small} implies
\begin{equation}\label{eq:classical-max-upgrade}
 H_{\max}^{\frac{\varepsilon}{\sqrt p}}(X|B)_\psi
 \le2\log\Gamma+\log\frac1p.
\end{equation}
The hypothesis gives $\varepsilon/\sqrt{p}\le\delta/8$, so the same upper bound holds
for $H_{\max}^{\delta/8}(X|B)_\psi$.
Combining this with~\eqref{eq:classical-prior-min} in
Fact~\ref{fact:classical-oneshot} proves~\eqref{eq:classical-upgrade}.
If the stated lower bound is negative, a one-message code suffices.
\end{proof}

\subsection{Exponential strong converse bound}
The Holevo information of the channel $\N$ is
\begin{equation}\label{eq:holevo}
 \chi(\N)=\sup_{\{q_x,\sigma_x\}}
 \left[H\!\left(\sum_xq_x\N(\sigma_x)\right)
          -\sum_xq_xH\bigl(\N(\sigma_x)\bigr)\right]
 =\sup_{\{q_x,\sigma_x\}} I(X:B),
\end{equation}
and its classical capacity is given by~\cite{Holevo,SW}
\begin{equation}\label{eq:capacity}
 C(\N)=\sup_{n\ge1}\frac1n\chi(\N^{\otimes n}).
\end{equation}
In particular,
$\chi(\N^{\otimes n})\le nC(\N)$ for every $n$.
We now follow the same route as in Section~\ref{sec:strongconverse}:
a slowly decaying success probability would yield a reliable code above
capacity. The required weak converse is the following.

\begin{fact}[Weak converse for classical communication {\cite{Holevo,SW}}]
\label{fact:classical-weak}
For every $r_0>C(\N)$, there are $0<\delta<1/2$ and an integer $n_0$
such that, for every $n\ge n_0$, a classical code for $\N^{\otimes n}$
with average success probability at least $1-\delta$ has message size
$M'$ satisfying $\log M'\le nr_0$.
\end{fact}

\begin{proof}[Proof of Theorem~\ref{thm:classical-exponential}]
Set $r=C(\N)+\gamma$. Choose $C(\N)<r_0<r$ and fix $\delta$ and $n_0$
from Fact~\ref{fact:classical-weak}. Set
$\kappa=\min\{(\dim B)^2,(\dim E)^2\}$ 
and 
choose $0<\theta<1/2$ so that $2h(\theta)+2\theta\log(2\kappa)<r-r_0$, and set
\begin{equation}\label{eq:classical-exponent}
  \alpha_\gamma=\min\left\{\frac{c\theta^2}{4},
               \frac{r-r_0-2h(\theta)-2\theta\log(2\kappa)}4\right\}>0.
\end{equation}
For sufficiently large $n$, $D=\lfloor\theta n\rfloor$ satisfies
$3\le D\le n/2$ and $D\ge\theta n/2$.
Lemma~\ref{lem:polynomial}, with identical factors, gives an approximation
of $\Pi^{\otimes n}$ with error at most $2^{-c\theta^2n/4}$ and
$\log\Gamma_{n,D}\le n\bigl[h(\theta)+\theta\log(2\kappa)\bigr]$.
Suppose a code of rate at least $r$ has $p_{{\rm succ},n}\ge2^{-\alpha_\gamma n}$.
Then, uniformly over such codes,
\begin{equation}\label{eq:classical-uniform-eta}
 \frac{2^{-cD^2/n}}{\sqrt{p_{{\rm succ},n}}}
 \le2^{-c\theta^2n/4+\alpha_\gamma n/2}
 \le2^{-c\theta^2n/8}\le\frac\delta8
\end{equation}
for all sufficiently large $n$.
Applying Theorem~\ref{thm:oneshot} directly to $\N^{\otimes n}$ therefore
produces a code with average success at least $1-\delta$ and
\begin{align}
 \log M'
 &\ge\log M_n+2\log p_{{\rm succ},n}-2\log\Gamma_{n,D}
                          -4\log\frac1\delta-16\\
 &\ge n\bigl[r-2\alpha_\gamma-2h(\theta)-2\theta\log(2\kappa)\bigr]
                          -4\log\frac1\delta-16\\
 &\ge n(r_0+2\alpha_\gamma)-4\log\frac1\delta-16.
\end{align}
For sufficiently large $n$ this exceeds $nr_0$, contradicting
Fact~\ref{fact:classical-weak}. Hence $p_{{\rm succ},n}<2^{-\alpha_\gamma n}$.
\end{proof}


\section{Strong converse for private capacity}\label{sec:private-capacity}\label{sec:private-setup}
We now consider private classical communication over a general quantum channel $\N$. An $M$-message code for private communication consists of input states
$\rho_m$ and a decoding POVM $\{\Lambda_{m}: m=1, \dots, M\}$.
We use $K_A$ for the message register and $K_B$ for
Bob's decoded label, and denote their joint state with Eve by
\begin{equation}\label{eq:measured-private-code}
 \omega_{K_AK_BE}
 =\frac1M\sum_{m,m'=1}^M\ketbra{m,m'}{m,m'}\otimes
 \Tr_B\!\left[\big(\sqrt{\Lambda_{m'}}\otimes\I_E\big)U\rho_mU^\dagger
                    \big(\sqrt{\Lambda_{m'}}\otimes\I_E\big)\right],
\end{equation}
where, as before, $U:A\to BE$ is the Stinespring isometry of the channel. 
The decoding success probability and the secrecy error are defined by
\begin{equation}\label{eq:operational-errors}
 p_{\rm succ}=\Pr\{K_A=K_B\},\qquad
 \delta_{\rm sec}
 =\frac12\min_{\sigma_E}\Big\|
     \omega_{K_AE}-\frac{\I_{K_A}}M\otimes\sigma_E\Big\|_1.
\end{equation}

We can also measure the success probability and secrecy jointly by the secret-key fidelity, namely the fidelity to an ideal shared secret key.
Let $\tau_M=M^{-1}\sum_m\ketbra{m,m}{m,m}$ be the ideal shared key. Then, we define
\begin{equation}\label{eq:private-fidelity}
 F_{\rm key}
 =\max_{\sigma_E}
       F(\omega_{K_AK_BE},\tau_M\otimes\sigma_E) 
\end{equation}
where the minimum and maximum range over normalized states $\sigma_E$ on $E$.
This quantity measures closeness to a uniform shared key independent of Eve. To connect the errors in~\eqref{eq:operational-errors} to~\eqref{eq:private-fidelity}, let $\omega'_{K_AK_BE}$ be obtained from $\omega_{K_AK_BE}$ by replacing Bob's key label by Alice's in every classical block, leaving Eve's conditional state unchanged: 
\begin{equation}\label{eq:corrected-private-code}
 \omega'_{K_AK_BE}
 =\frac1M\sum_{m,m'=1}^M\ketbra{m,m}{m,m}\otimes
 \Tr_B\!\left[(\sqrt{\Lambda_{m'}}\otimes\I_E)U\rho_mU^\dagger
                    (\sqrt{\Lambda_{m'}}\otimes\I_E)\right] .
\end{equation}
Comparing the blocks of the block-diagonal states $\omega$ and $\omega'$ gives
\begin{align}
 \tfrac12\norm{\omega-\omega'}_1=1-p_{\rm succ},\qquad
 \tfrac12\min_{\sigma_E}\norm{\omega'-\tau_M\otimes\sigma_E}_1
 =\delta_{\rm sec}.
\end{align}
The triangle inequality and $1-\sqrt{F(\rho,\sigma)}\le\tfrac12\norm{\rho-\sigma}_1$
therefore imply
\begin{equation}\label{eq:private-fidelity-lower}
 \boxed{\quad F_{\rm key}\ge(p_{\rm succ}-\delta_{\rm sec})_+^2,\quad}
 \qquad x_+=\max\{x,0\}.
\end{equation}

For $n$ uses of the memoryless channel, an $(n,M_n)$ code is a code
for $\N^{\otimes n}$.
Its rate is $\frac1n\log M_n$ bits per channel use. We write
$p_{{\rm succ},n}$ and $\delta_{{\rm sec},n}$ for the quantities
in~\eqref{eq:operational-errors} for this code.
A rate $r$ is achievable if there is a sequence of such codes with
$\liminf_{n\to\infty}\frac1n\log M_n\ge r$,
$p_{{\rm succ},n}\to1$, and $\delta_{{\rm sec},n}\to0$ as $n$ goes to infinity. The private capacity $P(\N)$ is the supremum of achievable rates.

The one-shot private information is
\begin{equation}\label{eq:private-information}
 P^{(1)}(\N)
 =\sup_{\{q_x,\rho_x\}}
       \bigl[I(X:B)-I(X:E)\bigr]
 =\sup_{\{q_x,\rho_x\}}
       \bigl[H(X|E)-H(X|B)\bigr],
\end{equation}
computed with respect to the output state
$\sum_xq_x\ketbra xx_X\otimes U\rho_xU^\dagger$. The private capacity is~\cite{Devetak}
\begin{equation}\label{eq:private-capacity}
 P(\N)=\sup_{n\ge1}\frac1nP^{(1)}(\N^{\otimes n}).
\end{equation}

\subsection{A coherent representation of the code}

This section is devoted to a coherent representation of the code, following the notation of~\cite{WTB}. 
For each input $\rho_m$, choose a purification
$\ket{\varrho_m}_{S_AA}$ and prepare
\begin{equation}\label{eq:coherent-input}
 \ket{\psi_0}_{K_AS_AA}
 =\frac1{\sqrt M}\sum_{m=1}^M\ket m_{K_A}\ket{\varrho_m}_{S_AA}.
\end{equation}
The purifying register $S_A$ is called Alice's shield. 
We assume throughout that $S_A$ in $\ket{\varrho_m}_{S_AA}$ includes a copy of the key $m$.\footnote{Simply replace $\ket{\varrho_m}$ with $\ket{\varrho_m}\otimes \ket m$.}

As in Section~\ref{sec:classical-capacity}, it is more convenient to
assume that Bob's decoding POVM is projective. To this end, for the original POVM
$\{\Lambda_{m}:\, m=1, \dots, M\}$, define the isometry
\begin{equation}\label{eq:private-decoder}
 J:B\longrightarrow K_BS_B=:\widehat B,\qquad
 J\ket b=\sum_{m=1}^M\ket{m}_{K_B} \otimes \big(\ket{m}\otimes \sqrt{\Lambda_{m}}\ket b\big)_{S_B}.
\end{equation}
Here, $S_B$ is called Bob's shield register and by construction it contains a copy of Bob's key.
The projections $\ketbra{m}{m}_{K_B}\otimes\I_{S_B}$ reproduce the
original POVM when pulled back through the isometry $J$. We may therefore replace
$\N$ by its composition with $J$, exactly as in the classical-capacity
argument. Nevertheless, we keep in mind that the channel's output is bipartite $\widehat B = K_BS_B$ and $K_B$ is Bob's key register containing his decoded message $m'$. See Figure~\ref{fig:classical-private-code}.

Recall that~\eqref{eq:coherent-input} encodes a purification of the channel input. 
Then, the resulting output state is
\begin{align}
 \ket{ \psi}_{K_AS_AK_BS_BE}
 &=\big(\I_{K_AS_A}\otimes(J\otimes\I_E)U\big)\ket{\psi_0}_{K_AS_AA}\\
 &=\frac1{\sqrt M}\sum_{m,m'=1}^M
      \ket{m,m'}_{K_AK_B}\otimes \ket{v_{mm'}}_{S_AS_BE}. \label{eq:coherent-private-output}
\end{align}
Here $\norm{v_{mm'}}^2$ is Bob's conditional probability of outcome $m'$
when message $m$ was sent. Since $S_A$ and $S_B$ contain copies of $K_A$ and $K_B$, respectively,  tracing out $S_A,S_B$ and measuring the two
keys recovers the original protocol discussed above. Indeed,~\eqref{eq:coherent-private-output} is a purification of~\eqref{eq:measured-private-code}. 

We emphasize that in the privacy-test construction and the blowing-up argument below,
we use $U, \Pi, \widetilde\Pi$ for their embeddings through $J$, with
receiver system $\widehat B$. 
This receiver isometry leaves
$P^{(1)}(\N)$ invariant, preserves the operator approximation error,
and does not increase $\|\widetilde\Pi\|_\pi$ across $\widehat B:E$.
Any decoder for the enlarged channel can be composed with $J$ to give
a decoder for the original channel.

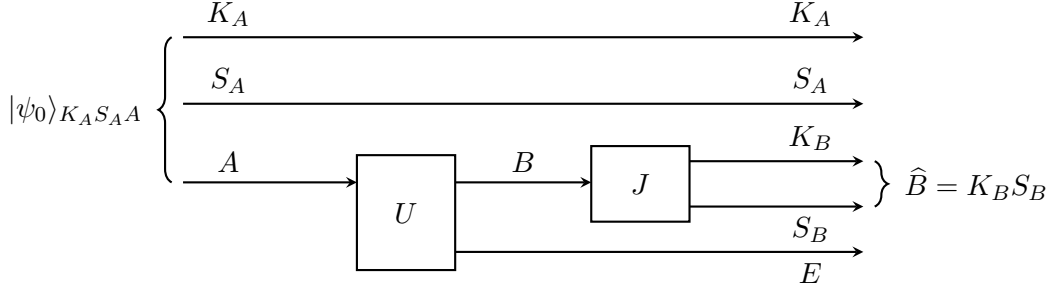
\begin{figure}[tbp]
\centering
\begin{tikzpicture}[x=1cm,y=0.8cm,>=stealth,thick]
 \draw[->] (0,2.4)--(9,2.4);
 \node[above] at (0.6,2.4) {$K_A$};
 \node[above] at (8.3,2.4) {$K_A$};
 \draw[->] (0,1.3)--(9,1.3);
 \node[above] at (0.6,1.3) {$S_A$};
 \node[above] at (8.3,1.3) {$S_A$};
 \draw[->] (0,0)--(2.3,0);
 \node[above] at (0.6,0) {$A$};
 \draw (2.3,-1.45) rectangle (3.6,0.45);
 \node at (2.95,-0.5) {$U$};
 \draw[->] (3.6,0)--(5.4,0);
 \node[above] at (4.5,0) {$B$};
 \draw (5.4,-0.65) rectangle (6.7,0.6);
 \node at (6.05,-0.025) {$J$};
 \draw[->] (6.7,0.35)--(9,0.35);
 \node[above] at (8.3,0.35) {$K_B$};
 \draw[->] (6.7,-0.4)--(9,-0.4);
 \node[below] at (8.3,-0.4) {$S_B$};
 \draw[->] (3.6,-1.15)--(9,-1.15);
 \node[below] at (8.3,-1.15) {$E$};
 \draw[decorate,decoration={brace,mirror,amplitude=5pt}]
 (-0.15,2.4)--(-0.15,0);
 \node[left] at (-0.4,1.2) {$\ket{\psi_0}_{K_AS_AA}$};
 \draw[decorate,decoration={brace,amplitude=5pt}]
 (9.15,0.35)--(9.15,-0.4);
 \node[right] at (9.4,-0.025) {$\widehat B=K_BS_B$};
\end{tikzpicture}
\caption{Coherent representation of classical communication through $\N$,
including the registers used for privacy. Alice prepares $\ket{\psi_0}_{K_AS_AA}$ where $K_A$ is Alice's key register, $S_A$ is her shield and $A$ is the input register. The Stinespring isometry $U$ outputs $B,E$.
Bob's isometry $J$ produces the key register $K_B$ and the shield $S_B$.
Measuring $K_A,K_B$ gives the sent and decoded messages. Eve holds $E$.}
\label{fig:classical-private-code}
\end{figure}

\subsection{Privacy-test projectors}\label{sec:choose-test}

Perfect secrecy, corresponding to $\delta_{\rm sec}=0$ in~\eqref{eq:operational-errors}, requires Eve's state to be independent of the key,
but the purifying shield systems $S_A, S_B$ may still be correlated with $E$. In a coherent
representation of an ideal private code, the conditional shield--environment
vectors are therefore different purifications of the same state of Eve.
By Uhlmann's theorem, these purifications are related by unitaries acting on
$S_AS_B$. This motivates the definition of a privacy-test projector as follows.

Since the shields contain a copy of the keys, tracing out the shields in
\eqref{eq:coherent-private-output} gives~\eqref{eq:measured-private-code}, i.e., $\psi_{K_AK_BE}=\Tr_{S_AS_B} \ketbra\psi\psi = \omega_{K_AK_BE}$. We therefore have
\begin{align}
 F_{\rm key} = \max_{\sigma_E}F(\psi_{K_AK_BE}, \tau_M\otimes\sigma_E) =\max_{\sigma_E}\left(\frac1M\sum_m
      \Big\|\sqrt{\sigma_E^{mm}}\sqrt{\sigma_E}\Big\|_1\right)^2,
\end{align}
where $\sigma_E^{mm'}=\Tr_{S_AS_B}\ketbra{v_{mm'}}{v_{mm'}}$ and we use the fact that both $\psi_{K_AK_BE}$ and $\tau_M\otimes\sigma_E$ are block diagonal.

Fix a normalized state $\sigma_E$ attaining the maximum in the definition of $F_{\rm key}$ and choose a purification $\ket\zeta_{S_AS_BE}$, adjoining a pure local ancilla to the shield if necessary.
By Uhlmann's theorem, there is a unitary $V_m$ acting on $S_AS_B$ such that
\begin{align}
 \braket{v_{mm}}{(V_m\otimes\I_E)|\zeta}
 =\Big\|\sqrt{\sigma_E^{mm}}\sqrt{\sigma_E}\Big\|_1.
\end{align}
Fixing these unitaries, define a privacy-test projector $Q$ by
\begin{equation}\label{eq:privacy-test}
 \boxed{\quad
 Q =\frac1M\sum_{m,m'=1}^M
       \ketbra{m,m}{m',m'}\otimes V_mV_{m'}^\dagger.
 \quad}
\end{equation}
Equivalently, $Q=LL^\dagger$ where $L$ is the isometry
\begin{align}\label{eq:Q-range-formula}
 L:S_AS_B\longrightarrow K_AK_BS_AS_B,\qquad
 L\ket z=\frac1{\sqrt M}\sum_m\ket{m,m}\otimes V_m\ket z.
\end{align}
Then $L^\dagger L=\I$ and $Q=LL^\dagger$. In particular,
$\operatorname{rank}Q=\dim(S_AS_B)$.

The unit vector
$\ket{\gamma_\sigma}:= (L\otimes \I_E)\ket{\zeta}=M^{-1/2}\sum_m\ket{m,m}\otimes (V_m\otimes\I_E)\ket\zeta$
belongs to $\ran(Q\otimes\I_E)$. Therefore
\begin{align}
 \bra\psi (Q\otimes \I_E)\ket\psi
 \ge |\braket{\gamma_\sigma}{\psi}|^2=\left(\frac1M\sum_m
      \Big\|\sqrt{\sigma_E^{mm}}\sqrt{\sigma_E}\Big\|_1\right)^2.
 \label{eq:uhlmann-private}
\end{align}
For the maximizing $\sigma_E$ fixed above, the resulting privacy-test projector $Q$ therefore satisfies  
\begin{equation}\label{eq:test-vs-fidelity}
 f_Q:=\bra\psi(Q\otimes\I_E)\ket\psi\ge F_{\rm key}.
\end{equation}
Inequality~\eqref{eq:test-vs-fidelity} is the privacy-test principle of \cite{WTB}.

\subsection{Quantum blowing-up lemma for private communication}\label{sec:private-oneshot}

Our quantum blowing-up lemma for private capacity is based on the following one-shot achievability result. 

\begin{fact}[One-shot private achievability {\cite[Theorem~1, Eq.~(6)]{RenesRenner}}]\label{fact:private-oneshot}
For any finite input ensemble $\{q_x,\rho_x\}$, let
$\omega_{XBE}=\sum_xq_x\ketbra xx\otimes U\rho_xU^\dagger$.
For $0<\delta<1$, there is an unassisted private code with
$1-p_{\rm succ}\le\delta$, $\delta_{\rm sec}\le\delta$, and message size
$M'$ satisfying
\begin{equation}\label{eq:private-achievability}
 \log M'\ge H_{\min}^{\delta/64}(X|E)_\omega
       -H_{\max}^{\delta/64}(X|B)_\omega
       -4\log\frac8\delta-16.
\end{equation}
\end{fact}

The cited result controls both errors for every message, hence also the average errors used here.

\begin{theorem}[Quantum blowing-up lemma for private communication]\label{thm:private-oneshot}
Suppose $\widetilde\Pi$ satisfies~\eqref{eq:oneshot-hyp}.
After the receiver isometry described above, let $Q$ be a privacy-test
projector, meaning a projector of the form~\eqref{eq:privacy-test}, with $M$-dimensional key registers, and let $\rho$ be a normalized
state on $K_AS_A\widehat BE$ with $\Pi\rho\Pi=\rho$.
Put $f=\Tr(\rho (Q\otimes \I_E))>0$ and fix $0<\delta<1/2$.
If $\varepsilon\le\delta\sqrt f/64$, there is a private code for the
same channel with decoding and secrecy errors both at most $\delta$ and
\begin{equation}\label{eq:private-oneshot}
 \boxed{\quad
 \log M'\ge\log M+2\log f-4\log\Gamma
                         -4\log\frac8\delta-16.
 \quad}
\end{equation}
Starting with an original private code, $f$ may be replaced throughout by $F_{\rm key}>0$.
\end{theorem}

\begin{proof}
\textit{Step 1: copy Alice's key into her shield.} As discussed above, it is desirable to assume that Alice's shield $S_A$ contains a copy of the key in $K_A$. To this end, let $C$ be a fresh
register and let $T\ket m_{K_A}=\ket m_{K_A}\ket m_C$. Extending operators with identities
on other registers, replace $\rho, Q$ by $T\rho T^\dagger, TQT^\dagger$ respectively.
The operator $TQT^\dagger$ is still a projection, and since $\Pi$ acts on registers disjoint from those on which $T$ acts, we have
\begin{align}
 &\Tr\big[(T\rho T^\dagger)(TQT^\dagger)\big]= \Tr[\rho Q]=f,\\
 &\Pi(T\rho T^\dagger)\Pi=T\rho T^\dagger,\\
 &\Pi(TQT^\dagger)\Pi=T(\Pi Q\Pi)T^\dagger.\label{eq:Pi-copy-commute}
\end{align}
In the following, for simplicity of presentation, we absorb $C$ into $S_A$ and keep the notation $\rho,Q$
for the modified operators. We also write $L$ for $TL$, so $Q=LL^\dagger$ as before.

The range formula for $Q$ remains valid
with the unitaries $V_m$ replaced by isometries from $S_A^0S_B$ to $S_AS_B$, where $S_A^0$ is the original shield space. Explicitly, each original $V_m$ is replaced by the map
$\ket z\mapsto\ket m_C\otimes V_m\ket z$. We denote these new isometries by $V_m$ as well; they satisfy $V_m^\dagger V_m=\I_{S_A^0S_B}$. Every vector in the range
of the modified $Q\otimes \I_E$ takes the form
\begin{align}
 \ket\theta_{K_AS_AK_BS_BE}=\frac1{\sqrt M}\sum_m\ket{m,m}_{K_AK_B}
                 (V_m\otimes\I_E)\ket\zeta_{S^0_AS_B E},
\end{align}
for some vector $\ket\zeta_{S^0_AS_B E}$ and contains the key copy in $S_A$.

\medskip
\noindent
\textit{Step 2: choose a vector in the channel image.}
Let $\lambda=\op{\Pi Q\Pi}$. The support assumption gives
$\lambda\ge\Tr(\rho Q)=f$.
Choose a unit eigenvector $\ket\psi$ of $\Pi Q\Pi$ with eigenvalue $\lambda$ and put
$\ket\theta=Q\ket\psi$. Since $Q$ is a projection,
\begin{equation}\label{eq:private-eigenvector}
 \Pi\ket\psi=\ket\psi,\qquad
 \Pi\ket\theta=\lambda\ket\psi,\qquad
 \norm{\ket\theta}^2=\lambda.
\end{equation}
Define
\begin{align}
 \ket{\widetilde\psi}=\frac{1}{\norm{\widetilde{\Pi}\ket\theta}}\widetilde{\Pi}\ket\theta.
\end{align}
We note that by~\eqref{eq:Pi-copy-commute}, all the vectors $\ket\psi, \ket{\widetilde\psi}$ and $\ket\theta$ lie in the copying
subspace and retain the key copy in $S_A$.

As in the vector construction in Theorem~\ref{thm:upgrade},
$\widetilde{\Pi}\ket\theta=\lambda\ket\psi+\ket v$ with
$\Pi\ket v=0$ and $\norm v\le\varepsilon\sqrt\lambda$.
Consequently
\begin{equation}\label{eq:private-distance}
 \norm{\widetilde{\Pi}\ket\theta}\ge\lambda,\qquad
 P(\ket\psi,\ket{\widetilde\psi})
 =\frac{\norm v}{\norm{\widetilde{\Pi}\ket\theta}}
 \le\frac{\varepsilon}{\sqrt\lambda}
 \le\frac{\varepsilon}{\sqrt f}.
\end{equation}

\medskip

\noindent
\textit{Step 3: establish the two marginal bounds.}
Since $\ket\theta$ is in the range of $Q=LL^\dagger$, we can write
\begin{align}
 \ket\theta=\frac1{\sqrt M}\sum_m\ket{m,m}_{K_AK_B}
                 (V_m\otimes\I_E)\ket\zeta_{S^0_AS_B E},
 \qquad \norm\zeta^2=\lambda.
\end{align}
Put $\sigma_E=\Tr_{S_A^0S_B}\ketbra\zeta\zeta$, so
$\Tr\sigma_E=\lambda$, and set
$\theta_{K_AS_AK_BS_BE}=\ketbra\theta\theta$. Then
\begin{equation}\label{eq:private-two-marginals}
 \theta_{K_AE}=\frac{\I_{K_A}}M\otimes\sigma_E,
 \qquad
 \theta_{K_AS_AE}\le\I_{K_A}\otimes\theta_{S_AE},
 \qquad \Tr\sigma_E=\Tr\theta_{S_AE}=\lambda.
\end{equation}
For the first identity, tracing out $K_B$ makes $K_A$ classical, and
tracing out the shields removes the isometries $V_m$. For the second inequality, tracing out
$\widehat B=K_BS_B$ likewise leaves a state classical on $K_A$.
Each diagonal block is positive and bounded by their sum $\theta_{S_AE}$,
which proves the inequality.

\medskip

\noindent
\textit{Step 4: lower-bound secrecy entropy.}
Let $\widetilde\psi=\ketbra{\widetilde\psi}
                                    {\widetilde\psi}$.
Use the decomposition~\eqref{eq:Kdecomposition} of $\widetilde{\Pi}$.
Applying Lemma~\ref{lem:test-marginal}\textup{(iii), then (i)} under the partial trace over
$\widehat B$, and then tracing out $S_A$, gives
\begin{equation}\label{eq:private-secrecy-domination}
 \widetilde\psi_{K_AE}
 \le\frac{\I_{K_A}}M\otimes \mu_E,
 \qquad
 \mu_E=\frac{\Gamma^2}{\norm{\widetilde\Pi\ket\theta}^2}\sum_i\beta_iY_i\sigma_EY_i^\dagger,
 \qquad \Tr \mu_E\le\frac{\Gamma^2\lambda}{\norm{\widetilde\Pi\ket\theta}^2}
                      \le\frac{\Gamma^2}{\lambda}.
\end{equation}
The marginal is classical on $K_A$ because the fresh key copy in
$S_A$ is traced out. 
Then the definition of min-entropy and~\eqref{eq:private-secrecy-domination} give
\begin{equation}\label{eq:private-secrecy-entropy}
 H_{\min}(K_A|E)_{\widetilde\psi}\ge\log M-\log\frac{\Gamma^2}{\lambda}.
\end{equation}

\medskip

\noindent
\textit{Step 5: upper-bound Bob's remaining uncertainty.}
Apply the same two parts of Lemma~\ref{lem:test-marginal}, this time retaining
$S_AE$.
The second marginal estimate in~\eqref{eq:private-two-marginals} gives
 \begin{align}
& \widetilde\psi_{K_AS_AE}\le\I_{K_A}\otimes \mu_{S_AE},
 \qquad \Tr \mu_{S_AE}\le\frac{\Gamma^2}{\lambda},\\
 &\mu_{S_AE}=\frac{\Gamma^2}{\norm{\widetilde\Pi\ket\theta}^2}
 \sum_i\beta_i(\I_{S_A}\otimes Y_i)\theta_{S_AE}(\I_{S_A}\otimes Y_i)^\dagger.
 \end{align}
Hence $H_{\min}(K_A|S_AE)_{\widetilde\psi}\ge-\log\frac{\Gamma^2}{\lambda}$.
Since $\widetilde \psi$ is pure, min/max-entropy duality gives
\begin{equation}\label{eq:private-reliability-entropy}
 H_{\max}(K_A|\widehat B)_{\widetilde\psi}
 =-H_{\min}(K_A|S_AE)_{\widetilde\psi}\le \log\frac{\Gamma^2}{\lambda}.
\end{equation}
We note that the marginal on $K_A\widehat B$ is also classical on $K_A$.

\medskip

\noindent
\textit{Step 6: apply one-shot private achievability.}
Let $\psi=\ketbra{\psi}{\psi}$ and identify the
measured key $K_A$ with a classical variable $X$.
Because $\Pi\ket\psi=\ket\psi$, the pullback
$\ket\xi=(\I_{K_AS_A}\otimes U^\dagger)\ket\psi$ is a unit input
vector. Measuring $K_A$ and tracing out $S_A$ gives an input ensemble
$\{q_x,\rho_x\}$.
Its corresponding channel output is precisely $\psi_{X\widehat B E}$.
By~\eqref{eq:private-distance} and contractivity, each corresponding
pair of marginals has purified distance at most $\varepsilon/\sqrt{f}\le\delta/64$.
Consequently, increasing the smoothing radius, the two bounds above give
\begin{align}\label{eq:private-smooth-entropies}
 H_{\min}^{\delta/64}(X|E)_\psi&\ge\log M-\log\frac{\Gamma^2}{\lambda},\\
 H_{\max}^{\delta/64}(X|\widehat B)_\psi&\le\log\frac{\Gamma^2}{\lambda}.
\end{align}
Thus Fact~\ref{fact:private-oneshot} yields
\begin{align}
 \log M'\ge\log M-2\log\frac{\Gamma^2}{\lambda}-4\log\frac8\delta-16
 \ge\log M+2\log f-4\log\Gamma-4\log\frac8\delta-16,
\end{align}
where the last inequality uses $\lambda\ge f$.
The receiver isometry can be absorbed into the new decoder, so this is
a code for the original channel. If the lower bound is negative, a
one-message code suffices. Finally, choosing the test
in~\eqref{eq:test-vs-fidelity} proves the assertion for $F_{\rm key}$.
\end{proof}

\subsection{Exponential strong converse bound}
We use Theorem~\ref{thm:private-oneshot} together with the following weak
converse, in the same way as for classical communication.

\begin{fact}[Weak converse for private communication {\cite[Sections~II--III]{Devetak}}]
\label{fact:private-weak}
For every $r_0>P(\N)$, there are $0<\delta<1/2$ and an integer $n_0$
such that every private code for $\N^{\otimes n}$, $n\ge n_0$, with
decoding and secrecy errors both at most $\delta$ has message size $M'$
satisfying $\log M'\le nr_0$.
\end{fact}

\begin{proof}[Proof of Theorem~\ref{thm:private-exponential}]
Set $r=P(\N)+\gamma$, choose $P(\N)<r_0<r$, and fix $\delta$ from
Fact~\ref{fact:private-weak}. Put
$\kappa=\min\{(\dim B)^2,(\dim E)^2\}$ for the physical channel.
Choose $0<\theta<1/2$ such that
$4h(\theta)+4\theta\log(2\kappa)<r-r_0$, and set
\begin{align}
  \alpha_\gamma=\min\left\{\frac{c\theta^2}{4},
 \frac{r-r_0-4h(\theta)-4\theta\log(2\kappa)}4\right\}>0.
\end{align}
For sufficiently large $n$, $D=\lfloor\theta n\rfloor$ satisfies
$3\le D\le n/2$ and $D\ge\theta n/2$.
Lemma~\ref{lem:polynomial} approximates the physical channel-image
projector $\Pi^{\otimes n}$ with error at most $2^{-c\theta^2n/4}$ and
$\log\Gamma_{n,D}\le n[h(\theta)+\theta\log(2\kappa)]$.
These bounds persist after the receiver isometry.
If $f_n\ge2^{-\alpha_\gamma n}$, then uniformly over the codes and tests,
\begin{align}
 \frac{2^{-cD^2/n}}{\sqrt{f_n}}
 \le2^{-c\theta^2n/4+\alpha_\gamma n/2}
 \le2^{-c\theta^2n/8}\le\frac\delta{64}
\end{align}
for all sufficiently large $n$.
Applying Theorem~\ref{thm:private-oneshot} to $\N^{\otimes n}$ gives a
private code with both errors at most $\delta$ and
\begin{align}
 \log M'
 &\ge\log M_n+2\log f_n-4\log\Gamma_{n,D}-4\log\frac8\delta-16\\
 &\ge n[r-2\alpha_\gamma-4h(\theta)-4\theta\log(2\kappa)]
                                  -4\log\frac8\delta-16\\
 &\ge n(r_0+2\alpha_\gamma)-4\log\frac8\delta-16.
\end{align}
For sufficiently large $n$ this exceeds $nr_0$, contradicting
Fact~\ref{fact:private-weak}. Therefore $f_n<2^{-\alpha_\gamma n}$.
Equation~\eqref{eq:test-vs-fidelity} gives the assertion for $F_{{\rm key},n}$.
\end{proof}

\begin{corollary}[Reliability--secrecy tradeoff]\label{cor:tradeoff}
At a rate at least $P(\N)+\gamma$, every sufficiently long private code
satisfies
\begin{equation}\label{eq:tradeoff}
 \boxed{\quad
 p_{{\rm succ},n}
 \le\delta_{{\rm sec},n}+2^{-\alpha_\gamma n/2}.
 \quad}
\end{equation}
Equivalently, with decoding error $\varepsilon_{{\rm dec},n}=1-p_{{\rm succ},n}$,
\begin{align}
 \varepsilon_{{\rm dec},n}+\delta_{{\rm sec},n}
 \ge1-2^{-\alpha_\gamma n/2}.
\end{align}
In particular, vanishing secrecy error forces decoding success to vanish.
\end{corollary}
\begin{proof}
Combine~\eqref{eq:private-fidelity-lower} with
Theorem~\ref{thm:private-exponential}.
\end{proof}

The exponential statement concerns secret-key fidelity. Equation~\eqref{eq:tradeoff}
also gives exponential decay of decoding success whenever the secrecy
error itself decays exponentially. If secrecy error merely tends to
zero, the conclusion for decoding success is convergence to zero.

\section{Conclusion}\label{sec:conclusion}
The strong converses in this work follow from a fully quantum version of the blowing-up strategy. A code whose fidelity is sufficiently large can be converted into a high-fidelity code. At a rate above capacity, a code with fidelity larger than a suitable exponential threshold would therefore yield a high-fidelity code that violates the weak converse.

The argument separates two ingredients. The one-shot amplification theorem depends only on the accuracy and projective norm of the projector approximation. The memoryless channel structure supplies such an approximation through a low-degree polynomial, but is not required by the conversion itself. This separates the coding argument from the approximation problem and suggests a way to investigate other channel families by studying their Stinespring image projectors. Understanding when approximations exist, and whether their projective norm has a more direct operational interpretation, are natural directions for further work.

The resulting exponents are not optimized. Sharper polynomial approximations or channel-specific bounds on the projective tensor norm could improve the quantitative tradeoff.

\paragraph*{Acknowledgements.}
M.T. acknowledges support from the National Research Foundation Investigatorship Award (NRF-NRFI10-2024-0006) and the National Research Foundation, Singapore, through the National Quantum Office, hosted by A*STAR, under its Centre for Quantum Technologies Funding Initiative (S24Q2d0009).

\paragraph*{Use of artificial intelligence.} OpenAI Codex with ChatGPT~6 Astra assisted with exploring and checking proof strategies, and with preparing the manuscript. The current exposition of the proofs is due to the authors. The authors are responsible for the mathematical claims, the final presentation, and any remaining errors.


\appendix
\section{Alternative proof of the c-q blowing-up lemma}\label{app:product-blowing}

\begin{proof}[Proof of Theorem~\ref{thm:geometric}]
Our proof is based on the low-degree polynomial construction from Section~\ref{sec:approximation}.  First suppose $3\le D\le n$. Choose a purification $\ket{\omega_i}_{B_iR_i}$ of $\rho_i$, with all $R_i$'s being copies of the
same space $R$, and set
\begin{align}
 \ket\Omega=\bigotimes_{i=1}^n\ket{\omega_i},\qquad
 \Pi_i=\ketbra{\omega_i}{\omega_i},\qquad
 \Pi_\Omega=\ketbra\Omega\Omega= \Pi_1\cdots \Pi_n.
\end{align}
Applying Lemma~\ref{lem:polynomial} directly to the projections $\Pi_i$
on the pairs $B_iR_i$ gives an operator $\widetilde{\Pi}$ satisfying
\begin{align}
 \widetilde{\Pi}\Pi_\Omega=\Pi_\Omega \widetilde{\Pi}=\Pi_\Omega,\qquad
 \op{\widetilde{\Pi}-\Pi_\Omega}\le 2^{-cD^2/n},
\end{align}
and $\widetilde{\Pi}$ is a sum of operators supported on at most $D$ pairs $B_iR_i$.

Put $\ket v=(P\otimes\I_{R^n})\ket\Omega$. Then
\begin{align}
 \norm{\ket v}^2=p,\qquad
 \Pi_\Omega\ket v=p\ket\Omega,\qquad
 \norm{(\I-\Pi_\Omega)\ket v}^2=p(1-p).
\end{align}
Every operator on pairs indexed by $S$ is a sum of terms
$X_{B_S}\otimes Y_{R_S}$. Since $\ket v\in\ran (P)\otimes R^n$, the
support property of $\widetilde{\Pi}$ implies
$\widetilde{\Pi}\ket v\in\ran \big(P^{[D]}\big)\otimes R^n$.
Writing $Q=P^{[D]}\otimes\I_{R^n}$, we therefore obtain
\begin{align}
 p\norm{(\I-Q)\ket\Omega}
 &=\norm{(\I-Q)(\Pi_\Omega-\widetilde{\Pi})\ket v}\\
 &\le \op{\widetilde{\Pi}-\Pi_\Omega}\,
          \norm{(\I-\Pi_\Omega)\ket v}
 \le 2^{-cD^2/n}\sqrt{p(1-p)}.
\end{align}
Squaring gives the stronger estimate
\begin{equation}\label{eq:stronger-geometric}
 \Tr\left[\rho\big(\I-P^{[D]}\big)\right]
 \le\frac{1-p}{p}\,2^{-2cD^2/n},
 \qquad 3\le D\le n,
\end{equation}
and hence~\eqref{eq:starting}.
For $D\le2$, decrease $c$ so that $c\le1$ and use
$\Tr[\rho(\I-P^{[D]})]\le1-p\le(4p)^{-1}$ and
$D^2/n\le2$. This proves~\eqref{eq:starting} in all cases.
\end{proof}

\end{document}